\documentclass{sig-alternate-per}
\usepackage{bm,bbm}
\usepackage{verbatim}
\usepackage{color}
\usepackage{tabularx}
\usepackage{multirow}
\usepackage{amsmath,amsfonts,amssymb}
\usepackage{graphicx,subfigure,epsfig}
\usepackage{wrapfig}
\usepackage{setspace}
\usepackage{url}
\usepackage{hyperref}

\def\argmin{{\rm argmin}}
\def\argmax{{\rm argmax}}

\def\qed{\quad \mbox{$\vcenter{\vbox{\hrule height.4pt
				\hbox{\vrule width.4pt height.9ex \kern.9ex \vrule
					width.4pt} \hrule height.4pt}}$}}
\begin{document}
	\newtheorem{theorem}{Theorem}
	\newtheorem{lemma}{Lemma}
	\newtheorem{assumption}{Assumption}
	\newtheorem{definition}{Definition}
	\newtheorem{corollary}{Corollary}
	\newtheorem{remark}{Remark}
	\newtheorem{proposition}{Proposition}
	\newtheorem{conjecture}{Conjecture}

\title{Fork and Join Queueing Networks with Heavy Tails: Scaling Dimension and Throughput Limit}

\numberofauthors{3}

\author{
	\alignauthor
	Yun Zeng\\
	\affaddr{The Ohio State University}\\
	\affaddr{1971 Neil Ave}\\
	\affaddr{Columbus, OH}\\
	\email{zeng.153@osu.edu}
	\alignauthor
	Jian Tan\\
	\affaddr{The Ohio State University}\\
	\affaddr{2015 Neil Ave}\\
	\affaddr{Columbus, OH}\\
	\email{tan.252@osu.edu}
	\alignauthor Cathy H. Xia\\
	\affaddr{The Ohio State University}\\
	\affaddr{1971 Neil Ave}\\
	\affaddr{Columbus, OH}\\
	\email{xia.52@osu.edu}
}

\maketitle

\begin{abstract}
	Parallel and distributed computing systems are foundational to the success of cloud computing and big data analytics.  
	These systems process computational workflows in a way that can be mathematically modeled by a fork-and-join queueing network with blocking (FJQN/B). 
	While engineering solutions have long been made to build and scale such systems, it is challenging to rigorously characterize their throughput performance at scale theoretically.
	What further complicates the study is the presence of heavy-tailed delays that have been widely documented therein. 
	To this end, we introduce two fundamental concepts for networks of arbitrary topology (scaling dimension and extended metric dimension) and utilize an infinite sequence of growing FJQN/Bs to study the throughput limit. The throughput is said to be scalable if the throughput limit infimum of the sequence is strictly positive as the network size grows to infinity. 
	We investigate throughput scalability by focusing on heavy-tailed service times that are regularly varying (with index $\alpha>1$) and featuring the network topology described by the two aforementioned dimensions. 
	In particular, we show that an infinite sequence of FJQN/Bs is throughput scalable if the extended metric dimension $<\alpha-1$ and only if the scaling dimension $\le \alpha-1$. These theoretical results provide new insights on the scalability of a rich class of FJQN/Bs with various structures, including tandem, lattice, hexagon, pyramid, tree, and fractals. 
\end{abstract}


\section{Introduction}
Parallel and distributed computing systems are foundational to the success of cloud computing and big data analytics, 
witnessed by the wide applications deployed on, e.g., 
Amazon AWS~\cite{aws}, Google Cloud~\cite{GoogleCloud}, Microsoft Azure~\cite{MicrosoftAzure}, IBM BlueMix~\cite{IBMBlueMix}. 
Numerous large-scale analytics have been developed over distributed servers to achieve high performance,
e.g., for DNA sequencing analysis~\cite{genomics} and for astronomical data analysis~\cite{astronomy}. Parallel and distributed computing also exhibits itself in wireless sensor and ad-hoc networks~\cite{Sivrikaya2004,adhoc}, in composite web services~\cite{Menasce2004}, in distributed stream computing~\cite{S4}, in distributed file systems~\cite{GFS}, in MapReduce/Hadoop computing frameworks~\cite{Dean2008MapReduce,Shvachko2010Hadoop}, and in end-system multicast~\cite{Baccelli2005TCP},  etc. 


The above parallel and distributed computing systems can be naturally modeled as Fork-Join Queueing Networks with Blocking (FJQN/Bs), see, e.g.,~\cite{Olvera-Cravioto2014,Osman2015,XiaSig} and a recent survey by \cite{Thomasian14}.
A fork operation corresponds to a job being separated into subtasks for parallel processing at different service stations. A join operation corresponds to outputs of parallel subtasks being aggregated together at a synchronization point. Multiple service stations and fork/join operations exist and form a network. Intermediate jobs and subtasks are queued in buffers. Services and fork/join operations can be blocked when related buffers are fully occupied. Due to such synchronization and blocking mechanism, exact analyses of FJQN/Bs can be challenging and possess high complexity. Much of the literature focuses on performance properties such as stability, duality, and comparison results, 
e.g.,~\cite{BLiu89, Bacelli, Dallery}, 
approximation or bounding techniques, e.g., \cite{Balsamo98,FJPEVA2015,Rizk2015,Rizk2016,Tancrez2013,Varma94}, or heavy traffic limits, e.g.,~\cite{Lu2015,Lu2017}. 

As the sizes of various parallel and distributed computing systems continue to grow, their throughput performance could degrade due to synchronization delays, processing time variations, or data storage, I/O, and bandwidth constraints. The problem has been well recognized in distributed stream processing \cite{Cherniack2003,Jain2006,XiaSig}, in end-system multicast \cite{Baccelli2005TCP,Bhattacharyya1999,Chaintreau2002TCP}, in wireless networks \cite{Fu2002,Jelenkovic2007}, in cloud computing~\cite{Dean2013}, and in many other distributed computing environments. 
One critical issue concerns {\it throughput scalability}: can we properly design a parallel and distributed processing system in massive scale so that the throughput performance can be sustained independent of the size? 
While practical engineering solutions have long been made to scale computing systems, the mathematical foundations toward understanding the throughput performance of ever-growing systems remain rudimentary. 
What further complicates the investigation is the presence of heavy-tailed processing times that have been widely observed in such systems~\cite{Haddad2010,GoogleTrace,Kavulya2010,Scharf2005,Tan2012,Wang2015}. 
These heavy-tailed processing times can cause extremal delays that directly impact the synchronization and bring down system throughput. But how to quantify the throughput degradation? What are the key factors to determine scalability?

To investigate the throughput limit, we utilize an infinite sequence of FJQN/Bs $\mathcal{N}=\{N_1,N_2,\dots,N_i,\dots\}$ to characterize the way the system grows. 
Each $N_i$ is a FJQN/B of finite size (in number of nodes) while the network size goes to infinity as $i\rightarrow\infty$. 
This sequence $\mathcal{N}$ is said to be {\it throughput scalable} if the limit infimum of the network throughput is strictly positive. 
This scalability problem has been studied under light-tailed service times in \cite{Zeng2016}, which shows that, 
\vspace{-0.1in}
\begin{equation}\label{eqn:level-degree-bounded}\vspace{-0.07in}
	\limsup_{i\rightarrow\infty}D_i<\infty ~\text{and}~\limsup_{i\rightarrow\infty}L^*_i<\infty
\end{equation}
is a necessary and sufficient condition for throughput scalability of FJQN/Bs,
where $D_i$ and $L^*_i$ represent respectively the network degree and the minimum level of $N_i$. 
But the scalability condition remains open when the service times are heavy tailed. 


In this paper, we focus on scalability of FJQN/Bs under heavy-tailed service times; 
see \cite{Sigman1999} for different types of heavy-tailed distributions. In particular, we focus on an important class 
where a random service time $\sigma$ is regularly varying with index $\alpha>1$. In this case, we have $\mathbb{E}\left[\sigma^{\beta}\right]<\infty$, for any $\beta<\alpha$, and $\mathbb{E}\left[\sigma^{\beta}\right]=\infty$, for any $\beta>\alpha$; for more details on regularly varying see \cite{RegularVariation-Bingham}. 
We examine conditions for FJQN/Bs to be throughput scalable. 
The networks are assumed to be connected, directed, acyclic, and homogeneous in buffer sizes and service time distributions. For the non-homogeneous cases, we can always bound the throughput by that of homogeneous networks using the monotonicity property (see \cite{Bacelli,Dallery2}).  

Different from the light-tailed counterpart,  we show that,  for the heavy-tailed scenarios, the throughput scalability of FJQN/Bs is further determined by the following two concepts of network dimension: the \emph{scaling dimension} and the \emph{extended metric dimension}. The scaling dimension is formally defined in Section \ref{sec:Scaling-Dim}. Briefly speaking, the scaling dimension is given by the ratio of log network size over log diameter as the network expands. One can interpret the scaling dimension as a metric to measure how fast network grows as a function of network size and diameter. In particular, if $\mathcal{N}$ converges to a connected infinite graph that is locally-finite, then the scaling dimension is in analog with the growth degree in geometric group theory \cite{Imrich1991,Seifter1991}, or the upper internal scaling dimension in Physics \cite{Nowotny1998,Requardt2006}. If $\mathcal{N}$ converges to a fractal, then the scaling dimension is in analog with the box counting dimension \cite{FractalGeometry}, or the Hausdorff dimension \cite{Durhuus2009,Kron2004}. The extended metric dimension is formally defined in Section \ref{sec:Metric-Dim}. The concept derives from a graph's metric dimension: the minimum cardinality of a basis that uniquely identifies every node by its distance to the basis (see e.g.~\cite{MetricDimensionSurvey}). The extended metric dimension is given by the minimum cardinality of a basis that identifies nodes up to a constant level as network expands. One can interpret the extended metric dimension as the least number of coordinates needed to describe the network viewed far away as it expands.

Our main result includes a necessary condition and a sufficient condition on throughput scalability of FJQN/Bs under heavy-tailed service times. 
\vspace{-0.0in}
\begin{theorem}\label{thm:main-result}
	Consider an infinite sequence of FJQN/Bs $
	\mathcal{N}=\{N_i\}_{i=1}^\infty$, where $N_i=(V_i,E_i)$ is a finite-sized FJQN/B with $|V_{i}|<\infty$, $\forall i\in \mathbb{Z}^+$, and $\limsup_{i\rightarrow\infty}|V_i|=\infty$. 
	The service times are i.i.d. regularly varying with index $\alpha\!>\!1$. 
	Under condition~\eqref{eqn:level-degree-bounded}, the sequence $\mathcal{N}$ is throughput scalable if the extended metric dimension $dim_{EM}(\mathcal{N})$ satisfies\vspace{-0.05in}
	\begin{equation}\vspace{-0.05in} \label{eqn:suf}
		dim_{EM}(\mathcal{N})<\alpha-1
	\end{equation}
	and only if the scaling dimension $dim_S(\mathcal{N})$ satisfies\vspace{-0.05in}
	\begin{equation}\vspace{-0.05in} \label{eqn:nec}
		dim_S(\mathcal{N})\le \alpha-1.
	\end{equation}
\end{theorem}


Theorem \ref{thm:main-result} reveals that Condition \eqref{eqn:level-degree-bounded} is not enough to address throughput scalability in heavy-tailed cases. We need additional conditions on 
network dimension to ensure that the growth degree of the networks is bounded by the heavy tail index of the service time distribution. 
This result provides new insights on the scalability of a rich class of FJQN/Bs under various structures, including tandem, lattice, hexagon, pyramid, tree, and fractals. 
Table \ref{tbl:summary} provides a list of network examples with scalability conditions in addition to Condition \eqref{eqn:level-degree-bounded}, which will be further discussed in Section~\ref{sec:app}.

\begin{table}[h!]
	\centering
	\caption{\label{tbl:summary}Examples with Scalability Conditions}
	\vspace{0.05in}
	\small
	\begin{tabular}{| p{0.15\columnwidth}<{\centering} | p{0.3\columnwidth} | p{0.2\columnwidth}<{\centering} | p{0.17\columnwidth}<{\centering} |}
		\hline \rule{0pt}{8pt} 
		\multirow{2}{*}{Name} &     \hspace{0.25in}\multirow{2}{*}{Structure} & \multicolumn{2}{c|}{Scalability Conditions}  \\
		\cline{3-4} \rule{0pt}{8pt} 
		&& Necessary & Sufficient\\
		\hline    \rule{0pt}{10pt}
		Tandem &  \raisebox{-0.2\height}{\includegraphics[width=0.25\columnwidth]{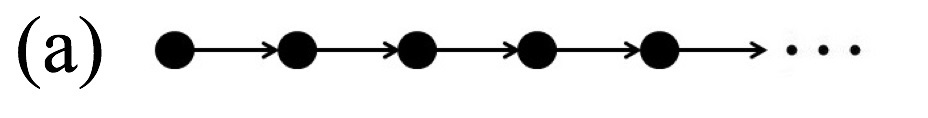}}& $\alpha\ge 2$ & $\alpha>2$ \\
		\hline  \rule{0pt}{25pt}
		Tandem-alike & \raisebox{-0.5\height}{\includegraphics[width=0.33\columnwidth]{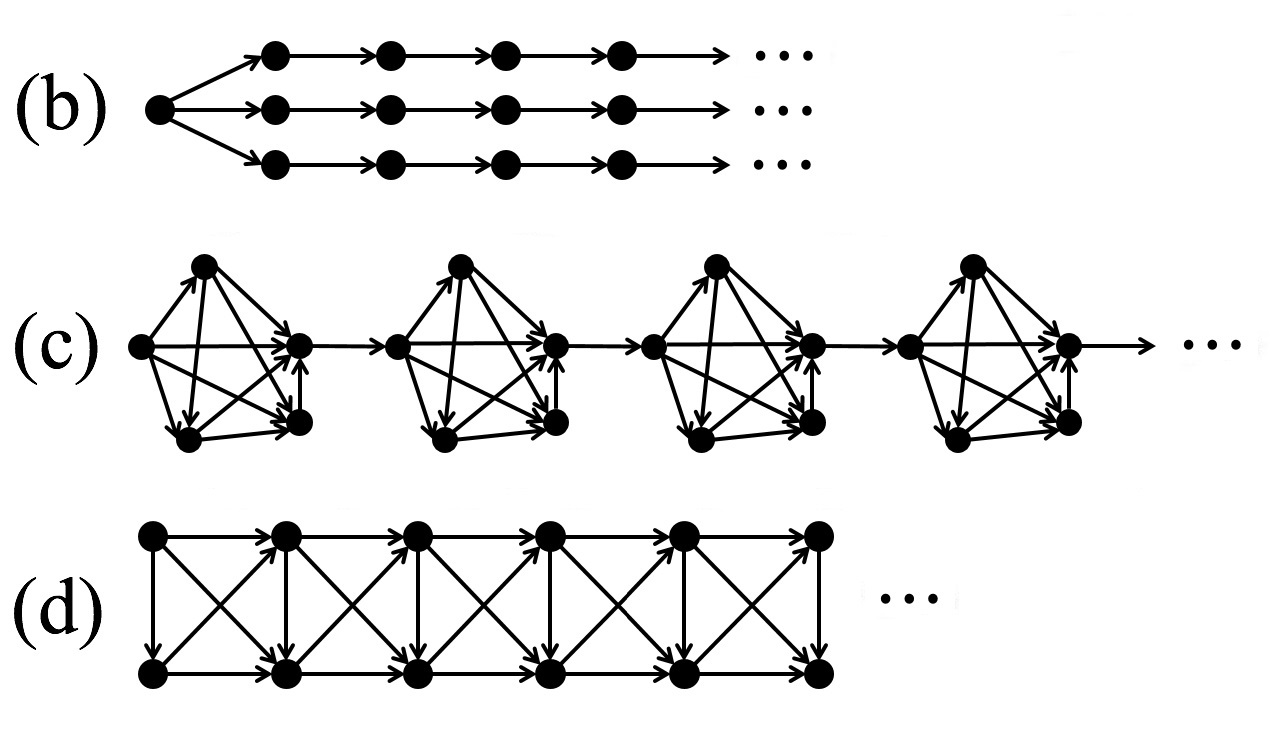}} & $\alpha\ge 2$ &  $\alpha>2$ \\
		\hline \rule{0pt}{32pt}
		$d$-D Lattice & \raisebox{-0.5\height}{\includegraphics[width=0.2\columnwidth]{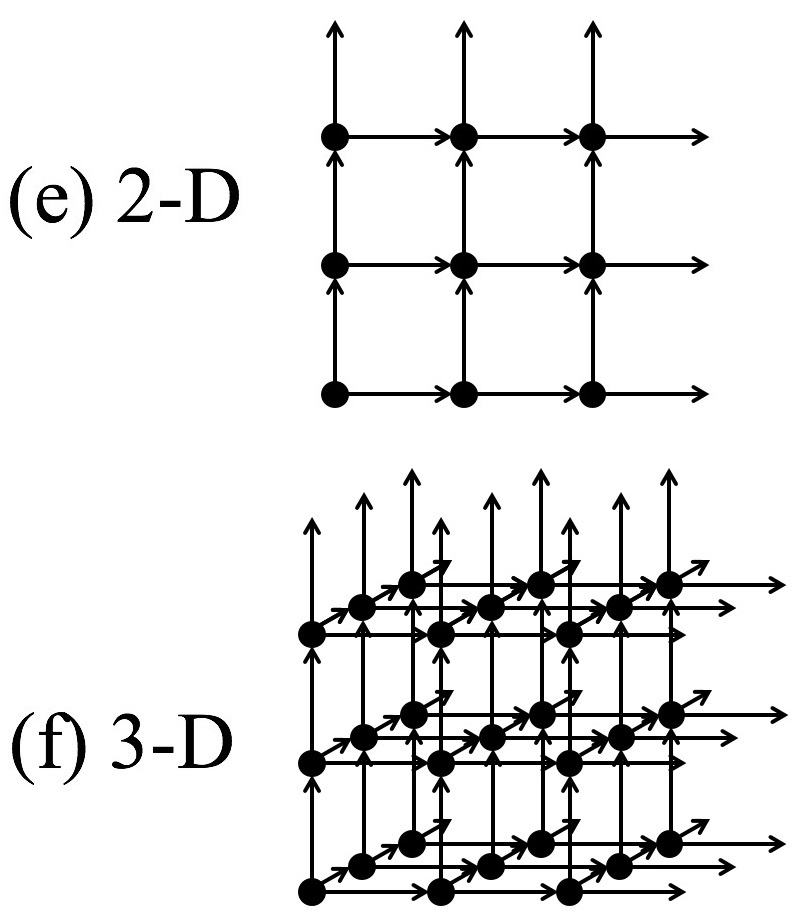}} \vspace{0.05in} & $\alpha\ge d+1$ & $\alpha>d+1$ \\
		\hline \rule{0pt}{24pt}
		Hexagon & \raisebox{-0.4\height}{\includegraphics[width=0.19\columnwidth]{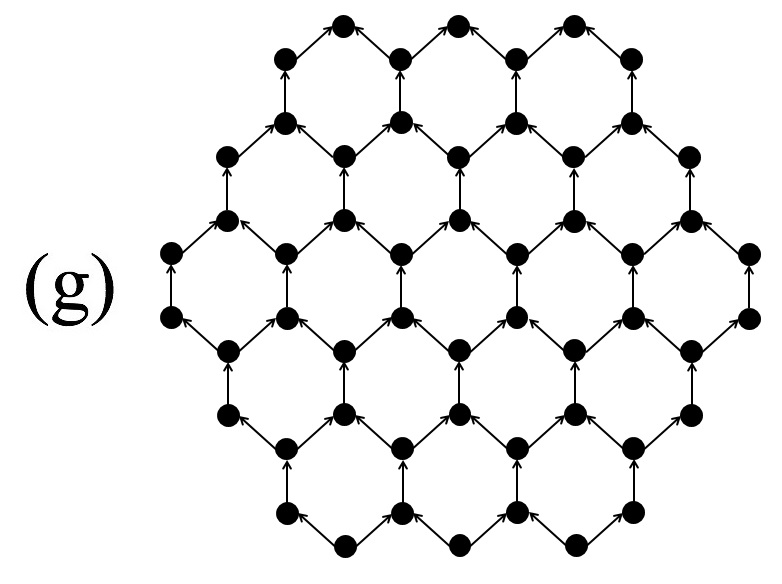}} \vspace{2pt} & $\alpha\ge 3$ & $\alpha>3$ \\
		\hline \rule{0pt}{15pt}
		Tetrahedron Pyramid & \raisebox{-0.75\height}{\includegraphics[width=0.2\columnwidth]{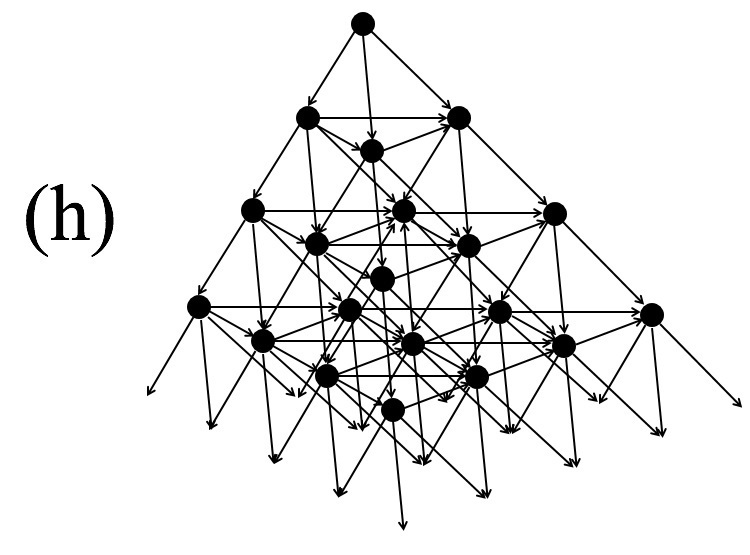}} & \vspace{0.01in} $\alpha\ge 4$ & \vspace{0.01in} $\alpha>4$ \\
		\hline \rule{0pt}{12pt}
		Sierpinski Triangle & \raisebox{-0.7\height}{\includegraphics[width=0.16\columnwidth]{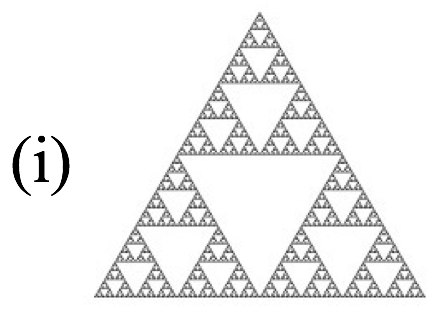}} \vspace{0.05in} & \vspace{0.01in} $\alpha\!\ge\! 1\!+\!\log_2 3$ & \vspace{0.01in} $\alpha>3$ \\
		\hline \rule{0pt}{20pt}
		Binary Tree & \raisebox{-0.5\height}{\includegraphics[width=0.18\columnwidth]{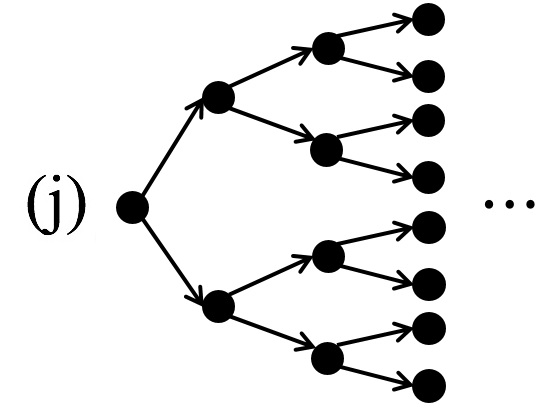}} \vspace{0.05in} & light-tailed & light-tailed \\
		\hline
	\end{tabular}
	\normalsize
\end{table}

\subsection{Contribution and Limitation}
We provide conditions for throughput scalability of general FJQN/Bs under heavy-tailed service times. Our contributions include: 
\begin{itemize}
	
	\item We introduce two important topological concepts on the dimension of an infinite sequence of FJQN/B:  scaling dimension and extended metric dimension. We demonstrate the relationship of the two dimensions, and  establish strong connections among the two dimensions, service time tails, and throughput limits.
	
	\item We propose a necessary condition on throughput scalability of FJQN/Bs depending on the scaling dimension and the service time tails. We show that, under the assumption that service times are i.i.d. regularly varying with index $\alpha$, a sequence of FJQN/Bs is not throughput scalable if the scaling dimension is strictly larger than $\alpha-1$. Thus, \eqref{eqn:nec} is necessary. The proof is based on last-passage percolation and extreme value theory.

	\item We propose a sufficient condition on throughput scalability of FJQN/Bs depending on the extended metric dimension and the service time tails. We show that, under the assumption that service times are i.i.d. with cdf $F_{\sigma}$  where
	\begin{equation} \label{eqn:FsigmaK}\vspace{-0.05in}
		\int_0^\infty \big(1-F_{\sigma}(x)\big)^{1/K}dx<\infty, 
	\end{equation} 
	for some finite~$K\in \mathbb{Z}^+$, 
	a sequence of FJQN/Bs is throughput scalable if the extended metric dimension is no larger than $K-1$. 
	When the service times are regularly varying with index~$\alpha$,  \eqref{eqn:FsigmaK} holds for any $K < \alpha$. Thus, \eqref{eqn:suf} is sufficient. The proof is based on mapping networks to lattices and bounding the throughput by growth of lattice animals \footnote{A lattice animal is a connect subset of points on a lattice; see \cite{Cox1993,Martin-LatticeAnimal} for the formal definition.}.
	
	\item We demonstrate that our proposed scalability conditions are almost tight (with only a marginal difference between ``$<$'' and ``$\le$'' in \eqref{eqn:suf} and \eqref{eqn:nec} of Theorem \ref{thm:main-result}) when the scaling dimension is an integer that equals the extended metric dimension.  This includes most of the commonly seen networks, including all the examples in Table \ref{tbl:summary} of tandem, lattice, hexagon, and tetrahedron pyramid networks. 
	
\end{itemize}


However, in the intriguing cases when the network converges to a fractal with a non-integral scaling dimension (e.g. the Sierpinski triangle in Table \ref{tbl:summary}), we observe that there exists a gap between the scaling dimension and the extended metric dimension. This leads to a non-trivial gap between the necessary and the sufficient conditions on throughput scalability. In general, $dim_{S}(\mathcal{N})\le dim_{EM}(\mathcal{N})$ as we establish in Lemma~\ref{lem:relation-S-le-EM}. 
We also conjecture that $dim_{EM}(\mathcal{N})\le \lceil dim_{S}(\mathcal{N})\rceil$,
which, if true, implies that the size of the gap is within $[0,1)$ for fractals and is marginal for common networks with integral scaling dimensions. 

To the best of our knowledge, this work is among the first attempts to  develop necessary and/or sufficient conditions of FJQN/Bs under heavy-tailed service times and establish the strong connections among the network dimensions, service time tails, and throughput limits. The results not only can cover FJQN/Bs with heavy-tailed service times but also can be applied to analysis of other types of networks or fractals such as social networks, electrical grid, Internet of Things, etc. 
Our investigation on the two network dimensions could also be of independent interest to a list of broad topics such as graph theory, geometric group theory, fractal geometry, and space-time physics.

\subsection{Related Work}
Previous studies on scalability of FJQN/Bs either focus on special network structures or assume  light-tailed service times. \cite{Martin} first shows that the throughput of a tandem queueing network is scalable, 
under condition \eqref{eqn:FsigmaK} for $K = 2$. 
\cite{Baccelli2005TCP} shows that the throughput of a one-to-many multicast tree 
is scalable, under light-tailed service times and bounded degree of the tree.  \cite{Chaintreau-sharpness} shows that the throughput of a pattern grid with dimension $d$
is scalable, if there exists a sharp vector of dimension $d$ and \eqref{eqn:FsigmaK} holds for $K = d$.
In \cite{Chaintreau}, the author gives an example that the throughput of a tree network is not scalable under heavy-tailed service times. 
For generally structured networks,  \cite{Zeng2016} proposes a necessary and sufficient condition for throughput scalability under light-tailed service times; 
\cite{XiaSig} presents necessary conditions for throughput scalability when service times are either light-tailed or of Pareto distributions. 
The question remains on how to characterize the throughput scalability of generally structured FJQN/Bs under heavy-tailed service times.

\medskip
The rest of the paper is organized as follows.  Section~\ref{sec:model} sets up the FJQN/B model and provides preliminary analysis. Section \ref{sec:scaling} introduces the concepts of scaling dimension and extended metric dimension, which is followed by a detailed discussion on applications in Section~\ref{sec:app}. Our main result is proved in Section \ref{sec:pf-main}. Section~\ref{sec:conclusion} concludes the paper.

\section{Model and Preliminaries}\label{sec:model}
\subsection{FJQN/B Model}
A Fork and Join Queueing Network with Blocking (FJQN/B), denoted by $N=(V, E; B)$, consists of a set of nodes $V$,  a set of directed arcs $E$, and a set of buffers $B$.  
Nodes represent servers, arcs represent routing of jobs. Associated with each arc $(u, v)$, there is a buffer of finite capacity representing intermediate storage of jobs between services. Arc $(u,v)$ is called an outgoing arc of node $u$ and an incoming arc of node $v$. Nodes with no incoming (outgoing) arcs are sources (sinks). The buffer on arc $(u,v)$ is called a downstream buffer of node $u$ and an upstream buffer of node $v$. One example of an FJQN/B is given in Figure \ref{fig:example}.

\begin{figure}
	\centering
	\includegraphics[width=0.6\columnwidth]{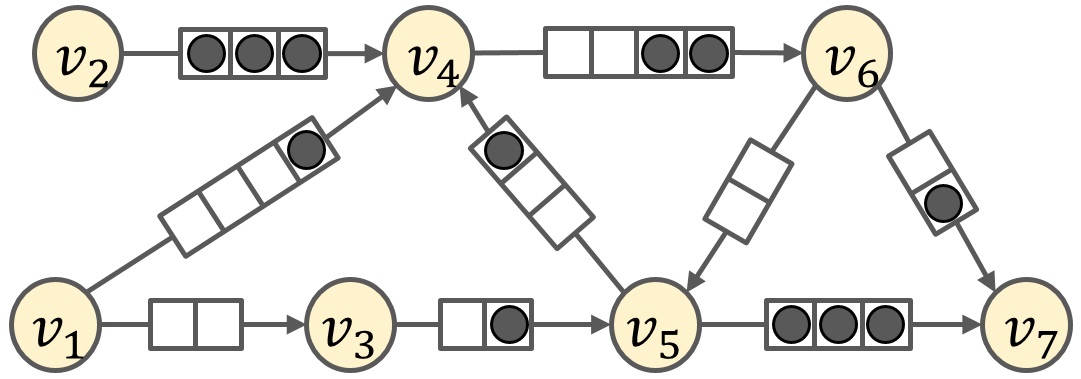}
	\caption{Example of FJQN/B. Blocked: $v_2,~v_5$; starved: $v_3,~v_5$.}
	\label{fig:example}
\end{figure}

Each node models a single server that serves incoming jobs according to the First Come First Serve (FCFS) policy. Services are conducted in a fork-join manner: each service consumes exactly one job from every upstream buffer and generates exactly one job to every downstream buffer. A server is starved (blocked) if one of the upstream (downstream) buffers is empty (full). Sources are never starved and sinks are never blocked. For example, in Figure \ref{fig:example}, servers on $v_2$ and $v_5$ are blocked; servers on $v_3$ and $v_5$ are starved. An idle server can schedule a service only when it is neither blocked nor starved. During the service, jobs remain in the buffers of incoming arcs. Upon completion of a service, one job is removed from each upstream buffer and one job is added to each downstream buffer. 
	Such mechanism, referred to as the blocking-before-service mechanism, can equivalently represent several other blocking mechanisms (see \cite{Dallery2}). Assume initially all servers are available. Such initial timing conditions have been shown independent of the throughput~\cite{Dallery}. 

For simplicity, we consider a homogeneous setting where all buffers are of constant size $b<\infty$ and are empty at time zero, and all service times are i.i.d. of the same distribution $F_{\sigma}$.
For the non-homogeneous cases, we can always bound the throughput by that of homogeneous networks using the monotonicity property (see \cite{Bacelli,Dallery2}).
We refer to the FJQN/B as $N=(V, E)$, and assume the underlying graph is connected, directed, and acyclic (DAG). In this paper, we focus on in particular the cases when $F_{\sigma}$ is regularly varying with index $\alpha$ defined as follows (see e.g. \cite{RegularVariation-Mikosch}) and we assume $\alpha>1$ .
\begin{definition}
	Distribution $F_{\sigma}$ is regularly varying with index 
	$\alpha$ if the tail distribution $\overline{F}_{\sigma}(x)=1-F_{\sigma}(x)$ satisfies
	\begin{equation}
		\lim_{x\rightarrow\infty}\frac{\overline{F}_{\sigma}(tx)}{\overline{F}_{\sigma}(x)}=t^{-\alpha},~~~\text{for all~}t>0.
	\end{equation} 
\end{definition}

Let $S_{m,v}(N)$ denote the $m$-th service time at node $v$, 
and  $T_{m,v}(N)$ the $m$-th service completion time at node $v$. 
The throughput at node $v\in V$ is defined as the average number of service completions in a unit time in the long run, 
namely, 
\begin{equation}\label{eqn:throughput-def-v}
	\theta_v(N)=\mathbb{E}\left[\left(\lim_{m\rightarrow \infty}\frac{T_{m,v}(N)}{m}\right)^{-1}\right].
\end{equation}
\normalsize
It is shown in \cite{Dallery} that when service times form jointly stationary and ergodic sequences (including i.i.d. as a special case): i) the limit in \eqref{eqn:throughput-def-v} exists; ii) the throughput at every node is identical; iii) the throughput of the network can be expressed as 
\begin{equation}\label{eqn:throughput-def}
	\theta(N)=\theta_v(N)=\left(\lim_{m\rightarrow \infty}\frac{\mathbb{E}\left[T_{m,v}(N)\right]}{m}\right)^{-1}.
\end{equation}

\subsection{Precedence Graph}\label{sec:precedence}
According to the block before service mechanism, $T_{m,v}(N)$ obeys the following recurrence equation (see e.g. \cite{Dallery,Zeng2016}):

\vspace{-0.1in}
\begin{equation}\label{eqn:rec}
	\begin{aligned}
		T_{m,v}(N)=S_{m,v}(N)+\max~
		&\{T_{m,u}(N)\big|(u,v)\in E\}\\
		&\cup\{T_{m-1,v}(N)\big|m\ge 1 \}\\
		&\cup\{T_{m-b,w}(N)\big|(v,w)\in E, m\ge b\}
	\end{aligned}
\end{equation}
with initial condition $T_{0,s}(N)=S_{0,s}(N)$, $s\in V^{source}$, where $V^{source}$ is the set of sources in $N$. The max term corresponds to the three conditions (the server is not starved; the previous job finishes process; the server is not blocked) under which the server on node $v$ can start processing job $m$. 

The recurrence equation \eqref{eqn:rec}
can be equivalently expressed as a last-passage percolation time on a directed graph in the following way (see e.g. \cite{Chaintreau,Martin,Zeng2016}); see~\cite{Martin-LPP} for a survey on last-passage percolation. Consider a {\it precedence graph} $\mathcal{G}=(\mathcal{V},\mathcal{E})$ which represents the collection of services and their precedence constraints as follows: 
\begin{eqnarray}
	&\hspace{-0.3in}&\bullet ~~  \mathcal{V}   =  (\{0\}\cup \mathbb{Z}^+) \times V    \label{eqn:G1}  \\\vspace{-0.1in}
	&\hspace{-0.3in}&\begin{aligned}\bullet  ~~ \mathcal{E}  =~&\mathcal{E}^I\cup \mathcal{E}^{II}\cup \mathcal{E}^{III},\\
		\text{~where~}
		&\mathcal{E}^I=\{(m,v)\!\rightarrow\!(m,u)\big|(u,v)\in E\}\\   
		&\mathcal{E}^{II}=\{(m,v)\!\rightarrow\!(m\!-\!1,v)\big|m\ge 1\}\\ &\mathcal{E}^{III}=\{(m,v)\!\rightarrow\!(m\!-\!b,w)\big|(v,w)\in E,m\ge b\}
	\end{aligned}    \label{eqn:G2}
	\\\vspace{-0.1in}
	&\hspace{-0.3in}&\bullet ~~ \mbox{weight $S_{m,v}(N)$ associated with each node $(m,v)$ in $\mathcal{V}$.} \nonumber
\end{eqnarray}

Let $\pi\!:\!(m,v)\!\leadsto\!(m',v')$ denote a directed simple path in $\mathcal{G}$ from node $(m,v)$ to node $(m',v')$. Let $Wei(\pi)$ denote the total weight of all nodes on $\pi$. By construction, we have the following lemma which is the foundation for us to graphically represent $T_{m,v}(N)$ using last-passage percolation. 
\begin{lemma}\label{lem:max-weight}
	$T_{(m,v)}(N)$ is given by the maximum weighted path from $(m,v)$ to $(0,v')$ for all $v'\in V$, i.e., 
	\begin{equation}\label{eqn:max-weight-lpp}
		T_{m,v}(N)=\max\{Wei(\pi)\big| \pi\!:\!(m,v)\leadsto (0,v'),v'\in V\}.
	\end{equation} 
\end{lemma}

	This Lemma together with the throughput definition \eqref{eqn:throughput-def} suggest that the asymptotic behavior of the max term in \eqref{eqn:max-weight-lpp} plays a critical role in determining the throughput limit of FJQN/Bs. This max term is subject to the structure of the precedence graph $\mathcal{G}=(\mathcal{V},\mathcal{E})$ which is essentially a representation of the structure of the fork-join network and its synchronization constraints. 
	The proof of our main result (see Section \ref{sec:pf-main} and Appendices \ref{apd:pf-nec},\ref{apd:pf-suf}) mainly focuses on charactering the asymptotic behavior of the max term in \eqref{eqn:max-weight-lpp} using the concepts of network dimensions introduced in Section \ref{sec:scaling}.

\subsection{Topological Concepts}\label{sec:topo-concepts}
We need the following topological concepts defined on a FJQN/B $N=(V,E)$ which is a DAG. 
\begin{definition}
	The network degree of $N$ is 
	\begin{equation}
		D(N)=\max\{deg(v)\big|v\in V\},
	\end{equation}
	where $deg(v)$ is the total number of arcs (in and out) connected to node $v$. 
\end{definition}

\begin{definition}\cite{Zeng2016}
	The minimum level of $N$ is
	\begin{equation}\vspace{-0.05in}
		L^*(N)=\mathop{\argmin}_{l:V\mapsto \mathbb{Z}}\{\max_{(i,j)\in E}\{l(j)-l(i)\}\},
	\end{equation}
	where $l:V\rightarrow \mathbb{Z}$ is a topological labelling \footnote{A topological labelling is a generalization of a topological sort which exists for every connected directed acyclic graph, see e.g. Kahn's algorithm \cite{Kahn}.} that maps each node $v\in V$ to an integer number $l(v)\in \mathbb{Z}$ such that $l(j)-l(i)\ge 1$, $\forall (i,j)\in E$. A topological labelling $l^*$ that achieves the minimum level is referred to as an optimal topological labelling of $N$ and is denoted by $l^*_{N}$.
\end{definition}

Let $G\!=\!(V,E)$ be the undirected counterpart of $N\!=\!(V,E)$. The diameter of $N$ is defined as follows. 

\begin{definition}
	The distance of two nodes $u,v$ in a graph~$G$, denoted as $dis(u,v)$, is the minimum number of arcs among all undirected paths connecting $u$ and $v$. 
\end{definition}

\begin{definition}
	The diameter of a graph $G$, denoted as~$\Delta(G)$, is the maximum of the distance of any pair of nodes in the graph, i.e. $\Delta(G)=\max\{dis(u,v)\big| \forall u,v\in V \}$. The diameter of a network $N$ is the diameter of its undirected counterpart~$G$, i.e. $\Delta(N)=\Delta(G)$. 
\end{definition}

\subsection{Throughput Scalability}
To discuss the throughput scalability of a FJQN/B as it grows in size, we introduce an infinite sequence of FJQN/Bs $\mathcal{N}=\{N_1,N_2,\dots,N_i,\dots\}$, 
where each $N_i=(V_i,E_i)$ is a finite-sized FJQN/B with $|V_{i}|<\infty$ and $\limsup_{i\rightarrow\infty}|V_i|=\infty$. 
That is, while each $N_i$ is a FJQN/B of finite size, the network sizes grow infinitely large along the sequence. Each network~$N_i$ is associated with network degree~$D_i$, minimum level $L^*_i$, and diameter $\Delta_i$. 
In addition, the service time distribution $F_\sigma$ and the buffer size $b$ are independent of $i$. We say that the sequence $\mathcal{N}$ is throughput scalable if the following condition holds.
\begin{definition}
	A sequence of FJQN/Bs $\mathcal{N}=\{N_i\}_{i=1}^\infty$ is \underline{throughput scalable} if and only if 
	\begin{equation}
		\liminf_{i\rightarrow \infty} \theta(N_i) > 0.
	\end{equation}
\end{definition}

It is shown in \cite{Zeng2016} that, under light-tailed service times, a sequence of FJQN/Bs is throughput scalable if and only if Condition~\eqref{eqn:level-degree-bounded} holds. In the next section, we demonstrate through examples that such condition is not enough to guarantee throughput scalability in heavy-tailed cases. 

\section{Preliminary Analysis}\label{sec:preliminary-heavy}
	In this section, we present preliminary analysis on three special examples and illustrate their scalability conditions under regularly varying service times. Such conditions are beyond Condition~\ref{eqn:level-degree-bounded} and depend on complicated characterizations of how the network scales, which motivates the propositions of the network dimensions in Section \ref{sec:scaling}. 
	
	\bigskip
	
	\noindent{\bf Tandem Network}: Consider a sequence of FJQN/Bs $\mathcal{N}=\{N_i\}_{i=1}^\infty$ where $N_i$ is a tandem network with a single source and $i$ downstream nodes. As $i \rightarrow \infty$, 
	the sequence converges to an infinite sequence of tandem queues, see Figure \ref{fig:tandem}. 
	\begin{figure}[ht!]
		\centering
		\includegraphics[width=0.32\columnwidth]{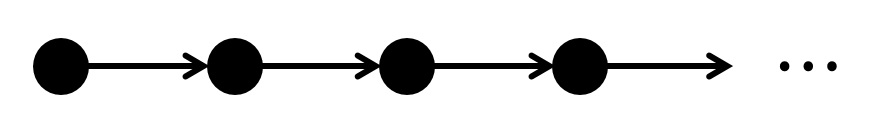}
		\caption{Tandem Network.}
		\label{fig:tandem}
	\end{figure}
	It is easily verified that $D_i=2$, $L^*_i=1$, and $\Delta_i= i$, for all~$N_i$, and Condition~\eqref{eqn:level-degree-bounded} is satisfied. Thus, the system is throughput scalable under light-tailed service times. 
	
	However, in heavy-tailed scenarios, if the service times are regularly varying with index $\alpha<2$, then the sequence will not be throughput scalable. To see this, consider the recurrence equation for $T_{m,v}(N_i)$ and the precedence graph~$\mathcal{G}_i$ for $N_i$. For simplicity, assume $b\!=\!1$ (the argument easily extends to the cases of any other constant buffer sizes). 
	For large $m$ as a multiple of $3 \Delta_i$, by dividing $[0, m]$ into equal intervals of length $3\Delta_i$, we can partition the precedence graph into 
	$\frac{m}{3 \Delta_i}$ layers as shown in Figure \ref{fig:tandem-lpp}. 	
	\begin{figure}[ht!]
		\centering
		\includegraphics[width=1\columnwidth]{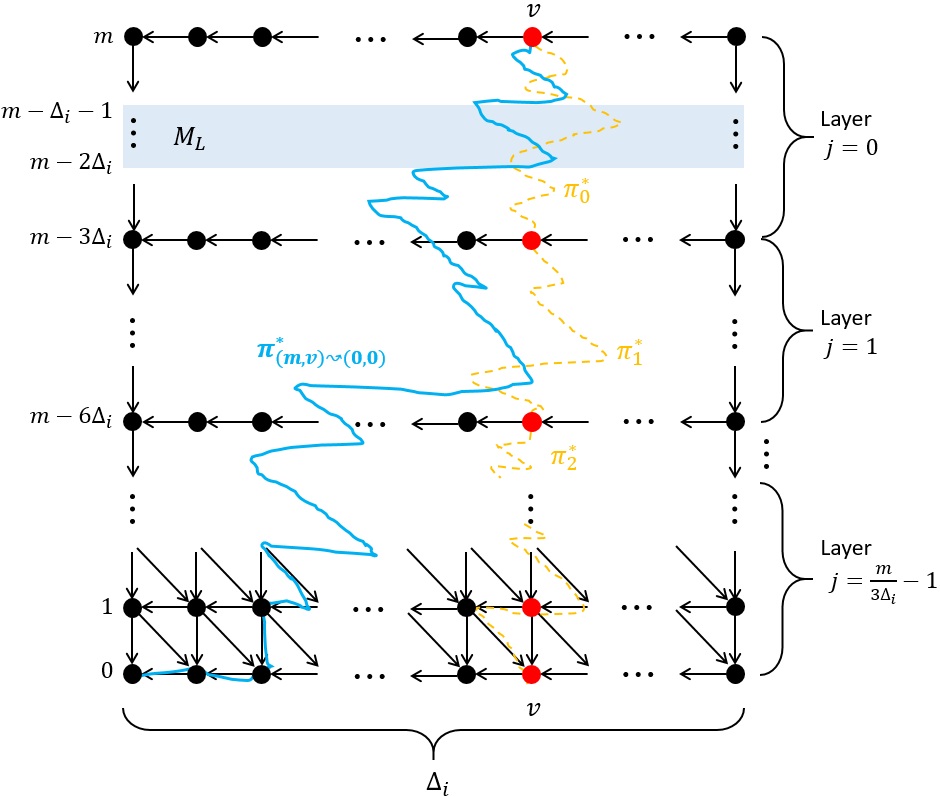}
		\caption{Last Passage Percolation on Tandem Network.}
		\label{fig:tandem-lpp}
	\end{figure}
	By super-additivity of the maximum weighted path (see Lemma \ref{lem:super-additive}), the weight of the maximum weighted path from $(m,v)$ to $(m\!-\!6\Delta_i,v)$ is bounded below by the weight of the maximum weighted path from $(m,v)$ to $(m\!-\!3\Delta_i,v)$ plus the weight of the maximum weighted path from $(m\!-\!3\Delta_i,v)$ to $(m\!-\!6\Delta_i,v)$ minus the weight on the duplicated point $(m\!-\!3\Delta_i,v)$. In essence, this action is to add an additional constraint on the path to go through the node $(m\!-\!3\Delta_i,v)$, which yields a lower bound on the maximum weighted path. Repeating the argument for all layers, we can bound $T_{m,v}(N_i)$, which is given by the maximum weighted path $\pi^*_{(m,v)\leadsto(0,0)}$ by definition, from below by the summation of $Wei(\pi^*_j)-S_{m-3j\Delta_i,v}$ for all $j\!=\!0,1,\dots,\frac{m}{3 \Delta_i}\!-\!1$, where $\pi^*_j$ denotes the maximum weighted path from $(m\!-\!3j\Delta_i,v)$ to $(m\!-\!3(j\!+\!1)\Delta_i,v)$.
	In layer $j\!=\!0$, a path from $(m,v)$ to $(m\!-\!3\Delta_i,v)$ may go through any node $(m',v')\!\in\! M_L$ where $M_L\!=\!\{(m', v'):  m'\!\in\![m\!-\!2\Delta_i,m\!-\!\Delta_i\!-\!1]$, $v'\!\in\! V_i\}$ represents the `middle layer' of layer $j\!=\!0$. The weight of the maximum weighted path in layer $j\!=\!0$ (excluding the weight on the starting point $(m,v)$) must be larger than the maximum weight of each individual node in the `middle layer', i.e. $Wei(\pi^*_0)\!-\!S_{m,v}\ge\max_{(m', v')\in M_L}\{S_{m',v'}(N_i)\}$. Repeat the argument for all layers and combine it with the lower bound on $T_{m,v}(N_i)$. As service times are i.i.d., 
	the expected value $\mathbb{E}\left[T_{m,v}(N_i)\right]$ is bounded below by $\frac{m}{3\Delta_i}\cdot \mathbb{E}\left[\max_{(m', v')\in M_L}\{S_{m',v'}(N_i)\}\right]$, where the total number of choices for $(m',v')$ is in the order of $\Delta_i|V_i|\sim i^2$. 
	By extreme value theory (see e.g. \cite[Theorem~3.3.7]{EVTEmbrechts}), the maximum of $n$ i.i.d. regularly varying (with index $\alpha$) random variables scaled by $n^{1/\alpha}L_1(n)$ converges weakly to a Fr{\'e}chet distribution, where $L_1$ is some slowly varying function. Hence, we can show that $\mathbb{E}\left[T_{m,v}(N_i)/m\right]$ grows at least in the order of $\frac{1}{3\Delta_i}\cdot \left(\Delta_i|V_i|\right)^{1/\alpha}\sim i^{2/\alpha-1}$ as $i\rightarrow\infty$. Then the throughput is at most in the order of $i^{1-2/\alpha}$ as $i\rightarrow\infty$. Thus, for regularly varying service times with index $\alpha<2$, the throughput decays to zero as the tandem network expands. In fact, the existence of the second moment of the service time distribution is known necessary for scalability of tandem networks \cite{Martin}. 
	
	\bigskip
	
	\noindent{\bf Binary Tree Network}: In comparison, consider a sequence $\mathcal{N}=\{N_i\}_{i=1}^\infty$ where~$N_i$ is a binary tree network with a root and $i$ layers (see Figure~\ref{fig:binarytree}). 
	\begin{figure}[ht!]
		\centering
		\includegraphics[width=0.35\columnwidth]{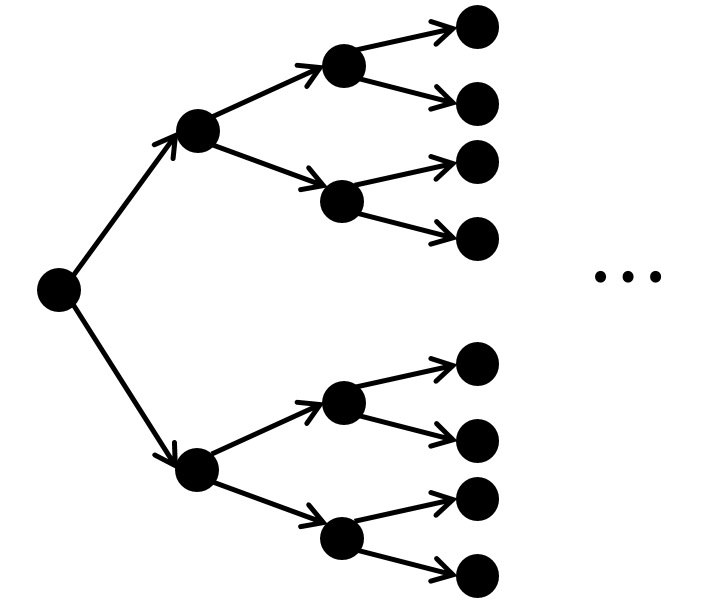}
		\caption{Binary Tree Network.}
		\label{fig:binarytree}
	\end{figure}
	It can be verified that $D_i=3$, $L^*_i=1$, and $\Delta_i= i$, for all~$N_i$, and Condition~\eqref{eqn:level-degree-bounded} is again satisfied, which is enough to guarantee throughput scalability if service times are light-tailed. However, in heavy-tailed cases where service times are regularly varying with any index $\alpha\in\mathbb{Z}^+$, the throughput will not be scalable. To see this, note that similar to the discussions in the tandem network case, $\mathbb{E}\left[T_{m,v}(N_i)\right]$ is bounded below by $\frac{m}{3\Delta_i}\cdot \mathbb{E}\left[\max_{(m', v')\in M_L}\{S_{m',v'}(N_i)\}\right]$, where the total number of choices for $(m',v')$ is in the order of $\Delta_i|V_i|\sim i2^i$. This leads to an exponential growth of $\mathbb{E}\left[T_{m,v}(N_i)/m\right]$ and hence the throughput decays exponentially fast to zero for any given index $\alpha\in \mathbb{Z}^+$ as the binary tree network expands. Similar discussions appear in \cite{Chaintreau}. 
	
	\bigskip
	
	From the above two examples, we observe that in addition to the network degree and the minimum level, the throughput limit under heavy-tailed service times depends on how fast the network grows, or essentially how many terms the middle layer $M_L$ contains. In addition, we need to identify the growth of the most critical part of the network. For instance, if the network consists of a tandem part and a binary tree part, then the asymptotic throughput will be dominated by the binary tree part and will decrease to zero under regularly varying service times with any index~$\alpha\in \mathbb{Z}^+$. These observations motivate us to introduce the metric of scaling dimension in Section~\ref{sec:Scaling-Dim}. 
	
	On the other hand, to provide sufficient conditions for throughput scalability, we need additional analyses to establish an upper bound on $\mathbb{E}\left[T_{m,v}(N_i)\right]$ and to derive a strictly positive lower bound on throughput limit. 
	Here we briefly demonstrate how to guarantee the scalability of regular lattice networks. 
	
	\bigskip
	
	\noindent{\bf Lattice Network}: Consider a sequence $\mathcal{N}=\{N_i\}_{i=1}^\infty$ where $N_i$ is a $d$-dimensional lattice network with $i$ nodes on each side (see e.g. Figure \ref{fig:lattice} for a 2-D lattice). 
	\begin{figure}
		\centering
		\includegraphics[width=0.3\columnwidth]{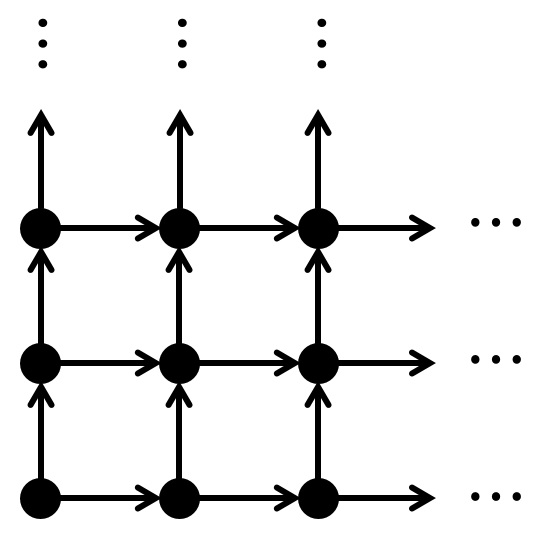}
		\caption{2-D Lattice Network.}
		\label{fig:lattice}
	\end{figure}
	%
	Again, assume $b\!=\!1$.
	Similar to the discussions above, we can show that $\mathbb{E}\left[T_{m,v}(N_i)/m\right]$ grows at least in the order of $\frac{1}{3\Delta_i}\cdot \left(\Delta_i|V_i|\right)^{1/\alpha}\sim i^{(d+1)/\alpha-1}$ and hence any $\alpha\!<\!d\!+\!1$ will make the sequence not scalable. Meanwhile, we can show that any $\alpha\!>\!d\!+\!1$ will ensure the scalability of the sequence. To see this, first divide $[0,m]$ into equal intervals of length $\Delta_i$, for large $m$ as a multiplication of~$\Delta_i$. The weight of the maximum weighted path from $(m,v)$ to $(m-2\Delta_i,v')$ for any $v,v'$ is bounded above by the weight of the maximum weighted path from $(m,v_1')$ to $(m\!-\!\Delta_i,v_1'')$, for all $v_1',v_1''$, plus the weight of the maximum weighted path from $(m\!-\!\Delta_i,v_2')$ to $(m\!-\!2\Delta_i,v_2'')$, for all $v_2',v_2''$. In essence, this action is to relax the constraint on the path to be connected between two adjacent layers by choosing $v_1''$ and $v_2'$ freely, which yields an upper bound on the maximum weighted path. Repeat the argument for all layers. As service times are i.i.d., the expected value $\mathbb{E}\left[T_{m,v}(N_i)\right]$ is bounded above by $\frac{m}{\Delta_i}\cdot \mathbb{E}\left[\max\{Wei(\pi)|\pi\!:\!(2\Delta_i,v')\!\leadsto\! (\Delta_i,v'')\},\forall v',v''\in V_i\right]$. Then we can show that the max term is bounded above by the weight of a greedy lattice animal of size $f(\Delta_i)$ on a $(d+1)$-dimensional lattice, where $f(\Delta_i)$ is a linear function of $\Delta_i$. Then the sufficient condition $\alpha>d+1$ follows from the result on the linear growth of lattice animals \cite{Martin-LatticeAnimal}. 
	
	\bigskip
	
	In the above example of lattice networks, the scalability is subject to the relationship between the service time tail index $\alpha$ and the dimension of the lattice $d$. This motivates us to further investigate such relationship in other networks. However, for other networks that are not as regular as lattices, we need to find a way to map the networks onto lattices so as to measure their dimensions. As the network expands in an arbitrary way along the sequence, we also need to develop a method to map the whole irregular sequence onto regular lattices and identify the lattice dimension that allows such mapping. Such dimension is introduce as the extended metric dimension in Section \ref{sec:Metric-Dim}. 
	
	Overall, we observe that the throughput scalability conditions in heavy-tailed cases are more complicated than that in light-tailed cases. Essentially, the necessary condition depends on the characterization of network growth as a function of network size and diameter; the sufficient condition depends on the construction of lattices on which we can embed the entire network sequence. These observations motivate us to propose two important geometrical concepts of network dimensions in the next section. 

\section{Characterization of Scaling}\label{sec:scaling}
Consider an infinite sequence of FJQN/Bs $
\mathcal{N}=\{N_i\}_{i=1}^\infty$, where each $N_i=(V_i,E_i)$ is a finite-sized FJQN/B with $|V_{i}|<\infty$ and $\limsup_{i\rightarrow\infty}|V_i|=\infty$. Let $G_i$ be the undirected counterpart of $N_i$. 
Each network~$N_i$ is associated with network degree $D_i$, minimum level $L^*_i$, and diameter~$\Delta_i$. The following lemma is immediate as a graph with bounded degree and diameter must have bounded size. 

\begin{lemma}\label{lem:degree-diameter-scaling}
	For an infinite sequence of FJQN/Bs $
	\mathcal{N}=\{N_i\}_{i=1}^\infty$, if $\limsup_{i\rightarrow\infty}|V_i|=\infty$, then we must have either $\limsup_{i\rightarrow\infty} D_i=\infty$, or $\limsup_{i\rightarrow\infty} \Delta_i=\infty$, or both. 
\end{lemma}
The condition $\limsup_{i\rightarrow\infty} D_i<\infty$ is shown necessary for throughput scalability of FJQN/Bs with light-tailed service times \cite{XiaSig,Zeng2016}. The same argument holds in our heavy-tailed service time settings. Thus, we focus on the cases where $\limsup_{i\rightarrow\infty} D_i\!<\!\infty$ and $\limsup_{i\rightarrow\infty} \Delta_i\!=\!\infty$. 

In the rest of this section, we first introduce the scaling dimension as a way to characterize the growth of the most critical part of the sequence $\mathcal{N}$ by a function of network size and diameter. Then we introduce the extended metric dimension as a way to map networks to lattices. The relationship between these two dimensions is further explored. 

\subsection{Scaling Dimension}\label{sec:Scaling-Dim}
As discussed in Section \ref{sec:preliminary-heavy} through the tandem and the binary tree examples, we need to characterize how fast the network grows so as to investigate throughput scalability. To this end, we propose the scaling dimension as follows. The scaling dimension determines the throughput upper bound and is used to provide a necessary condition on throughput scalability of FJQN/Bs in our main result.

\begin{definition}\label{def:scaling-dim}
	Consider an infinite sequence of FJQN/Bs $
	\mathcal{N}\!=\!\{N_i\}_{i=1}^\infty$ under Condition \ref{eqn:level-degree-bounded}. 
	Let $\Omega\left(\mathcal{\bar{I}},\mathcal{\bar{N}}\right)$ be the collection of $\left(\mathcal{\bar{I}},\mathcal{\bar{N}}\right)$ satisfying the following: 
	
	1) $\mathcal{\bar{I}}\!=\!\{i_n\}_{n=1}^\infty$ is a sequence of strictly increasing natural numbers;
	
	2) $\mathcal{\bar{N}}\!=\!\{\bar{N}_{i_n}\}_{n=1}^\infty$, where $\bar{N}_{i_n}\!=\!(\bar{V}_{i_n},\bar{E}_{i_n})$ is a connected subnetwork of $N_{i_n}$ with $\bar{V}_{n}\subseteq V_{i_n}$ and $\bar{E}_{n}\subseteq E_{i_n}$;
	
	3) $\Delta(\bar{N}_{i_n})\rightarrow\infty$ as $n\rightarrow\infty$.
	
	\noindent The \underline{scaling dimension} of the sequence $
	\mathcal{N}$ is defined as
	\begin{equation}
		dim_S(\mathcal{N})=\sup_{\left(\mathcal{\bar{I}},\mathcal{\bar{N}}\right)\in \Omega\left(\mathcal{\bar{I}},\mathcal{\bar{N}}\right)}\left\{
		\limsup_{n\rightarrow\infty}\frac{\log |\bar{V}_{i_n}|}{\log \Delta(\bar{N}_{i_n})}
		\right\}.
	\end{equation}
\end{definition}	

In words, the scaling dimension of a sequence of FJQN/Bs is defined by the limsup ratio of log network size over log diameter among all subsequences and subnetworks such that the diameter goes to infinity. This characterizes the growing speed of the most critical part of the sequence. The following lemmas demonstrate network examples and their scaling dimensions. Note that the scaling dimension does not depend on the direction of arcs in the networks. 

\begin{remark}\label{rmk:scal-dim-tandem}
	If $\mathcal{N}$ converges to an infinite tandem network (see Figure~\ref{fig:tandem}) where $N_i$ has a source and $i$ downstream nodes, then we have $|V_i|=i+1$, $\Delta_i=i$, and $dim_S(\mathcal{N})=1$. 
\end{remark}
\begin{remark}\label{rmk:scal-dim-lattice}
	If $\mathcal{N}$ converges to a $d$-dimensional lattice where $N_i$ has $i$ nodes on each side (see e.g. Figure \ref{fig:lattice} for a 2-D lattice), then we have $|V_i|= i^d$ and $\Delta_i=di-d$. In this case, the scaling dimension $dim_S(\mathcal{N})$ is an integer and equals the lattice dimension $d$. 
\end{remark}
\begin{remark}\label{rmk:scal-dim-hexagon}
	If $\mathcal{N}$ converges to an infinite hexagon network where $N_i$ has $i$ hexagons on each side (see Figure \ref{fig:hexagon}), then we have $|V_i|=6i^2$, $\Delta_i=4i-1$, and $dim_S(\mathcal{N})=2$. 
\end{remark}
\begin{figure}[ht!]
	\centering
	\includegraphics[width=0.6\columnwidth]{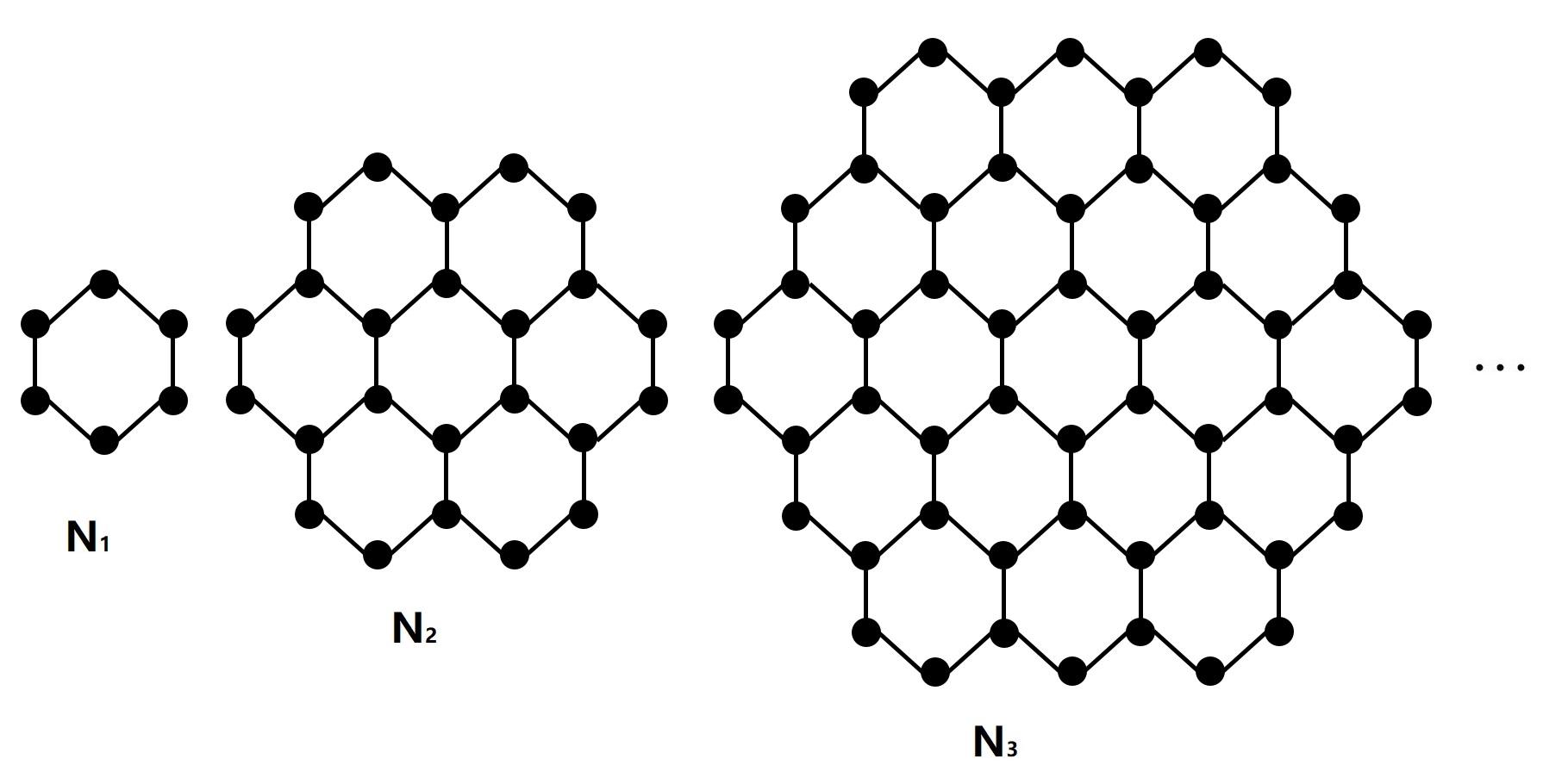}
	\caption{Hexagon Network.}
	\label{fig:hexagon}
\end{figure}
\begin{remark}\label{rmk:scal-dim-tetrahedron}
	If $\mathcal{N}$ converges to an infinite tetrahedron pyramid network where $N_i$ has $i$ layers (see Figure \ref{fig:tetrahedron}), then we have $|V_i|= \frac{1}{6}i^3+i^2+\frac{11}{6}i+1$, $\Delta_i=i$, and $dim_S(\mathcal{N})=3$. 
\end{remark}
\begin{figure}[ht!]
	\centering
	\includegraphics[width=0.35\columnwidth]{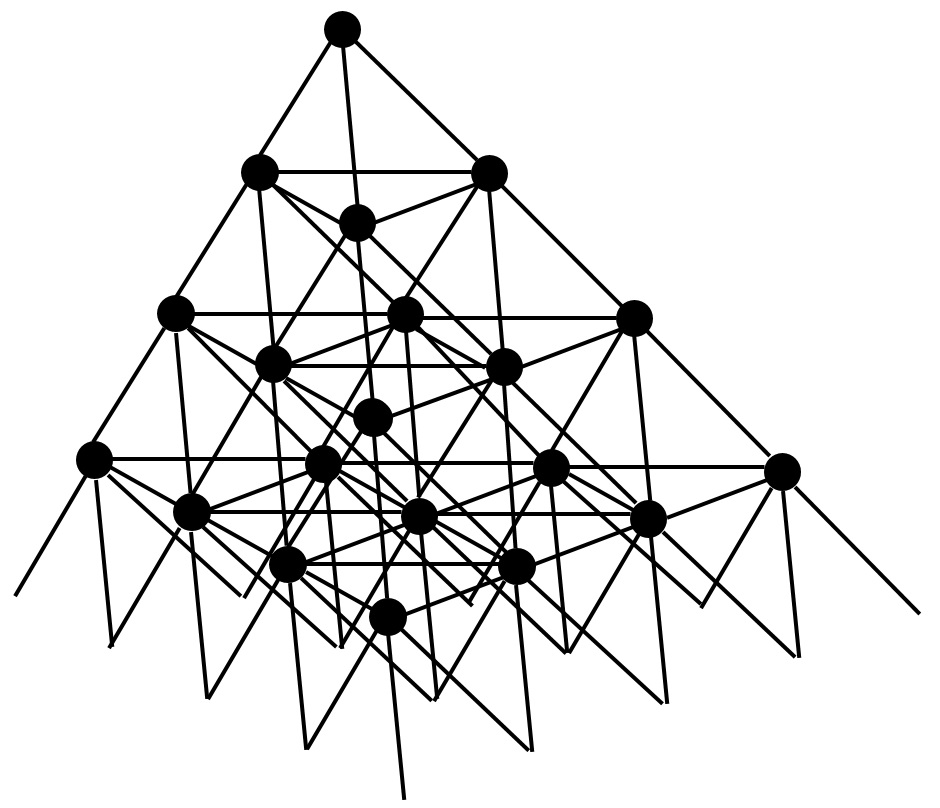}
	\vspace{-0.1in}
	\caption{Tetrahedron Pyramid Network.}
	\vspace{-0.1in}
	\label{fig:tetrahedron}
\end{figure}
\begin{remark}\label{rmk:scal-dim-triangle}
	If $\mathcal{N}$ converges to a Sierpinski triangle in a way shown in Figure \ref{fig:Sierpinski}, then we have $|V_i|=3^{i-1}\cdot \frac{3}{2}+\frac{3}{2}$ and $\Delta_i=2^{i-1}$. The scaling dimension $dim_S(\mathcal{N})$ is equal to $\log_2 3\approx 1.585$, which is the Hausdorff dimension (see \cite{FractalGeometry}). 
\end{remark}
\begin{figure}[ht!]
	\centering
	\includegraphics[width=0.8\columnwidth]{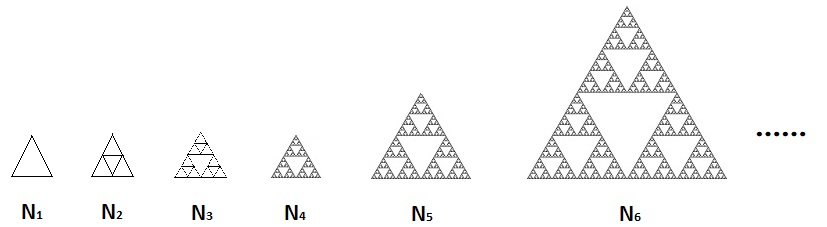}
	\vspace{-0.1in}
	\caption{Sierpinski Triangle.}
	\vspace{-0.1in}
	\label{fig:Sierpinski}
\end{figure}
\vspace{-0.1in}
\begin{remark}\label{rmk:scal-dim-binarytree}
	If $\mathcal{N}$ converges to a binary tree (see Figure~\ref{fig:binarytree}) where $N_i$ is a binary tree network with a root and $i$ layers, then we have $|V_i|=2^{i+1}-1$, $\Delta_i=i$, and hence $dim_S(\mathcal{N})=\infty$. In general, a tree with node degree $\ge 2$ grows exponentially fast (network size is an exponential function of the diameter) which makes the scaling dimension infinite. 
\end{remark}
\begin{remark}\label{rmk:scal-dim-tandem-bitree}
	In cases where $\mathcal{N}$ does not converge to a regular pattern or even does not converge, the construction of supremum over all subsequences and subnetworks enforces that we focus on the most critical part as the network size expands. For instance, consider a sequence $\mathcal{N}=\{N_i\}_{i=1}^\infty$ where $N_i$ is a tandem network of $2^i$ arcs (shown in Figure~\ref{fig:1D-tandem}) if $i$ is odd, and $N_i$ is a tandem network of $2^i$ arcs attached by a binary tree of $i$ layers (shown in Figure \ref{fig:1D-tandem-binarytree}) if $i$ is even. For the sequence, we have
	\begin{equation}
		\limsup_{i\rightarrow\infty}\frac{\log |V_{i}|}{\log \Delta(N_{i})}\!=\!\limsup_{i\rightarrow\infty}\frac{\log \left(2^i\!+\!1\!+\!\left(2^{i+1}\!-\!2\right)\!\cdot\! \mathbbm{1}_{i \text{~even}}\right)}{\log \left(2^i+i\cdot \mathbbm{1}_{i \text{~even}}\right)}\!=\! 1,
	\end{equation}
	where $\mathbbm{1}$ is the indicator function. However, consider a subsequence where $i_n=2n$ and let $\bar{N}_{i_n}=(\bar{V}_{i_n},\bar{E}_{i_n})$ be the binary tree part of $i_n$ layers. We have
	\begin{equation}
		\limsup_{n\rightarrow\infty}\frac{\log |\bar{V}_{i_n}|}{\log \Delta(\bar{N}_{i_n})}=
		\limsup_{n\rightarrow\infty}\frac{\log \left(2^{2n+1}\!-\!1\right)}{\log (4n)}= \infty.
	\end{equation}
	Hence, the scaling dimension of the sequence is infinite, which is determined by the binary tree part. 
\end{remark}
\begin{figure}[ht!]
	\centering
	\hspace{-1.3in}
	\includegraphics[width=0.4\columnwidth]{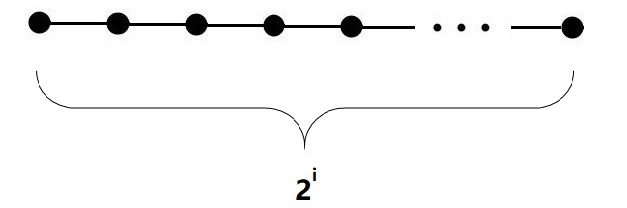}
	\caption{Tandem Network.}
	\label{fig:1D-tandem}
\end{figure}
\begin{figure}[ht!]
	\centering
	\includegraphics[width=0.6\columnwidth]{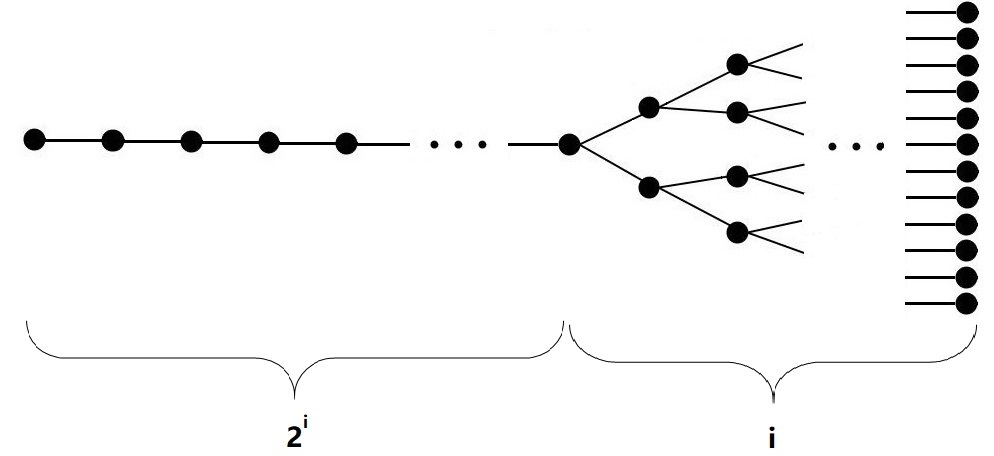}
	\caption{Tandem + Binary Tree Network.}
	\label{fig:1D-tandem-binarytree}
\end{figure}

\subsection{Extended Metric Dimension}\label{sec:Metric-Dim}
	As illustrated by the lattice network example in Section \ref{sec:preliminary-heavy}, we look for methods to map networks onto regular lattices and to identify the lattice dimension that allows such mapping. 
	Toward this purpose, we introduce in this section the traditional concept of a graph's metric dimension \cite{Slater1975,Harary1976}, which allows a one-to-one mapping from nodes on a network to points on a lattice. Then we propose the extended metric dimension as a method to map a sequence of networks onto lattices such that, for each network regardless of its size, constantly many nodes can share the same position on a lattice. The extended metric dimension determines the throughput lower bound and is used to provide a sufficient condition on throughput scalability of FJQN/Bs in our main result. Such extension from the metric dimension to the extended metric dimension is critical to characterize how the network scales along the sequence (as discussed in Remarks \ref{rmk:extm-dim-cycle},\ref{rmk:extm-dim-complete-tandem}, and \ref{rmk:extm-dim-ladder}) and to diminish the gap between the necessary and the sufficient conditions in our main result.

Let $G=(V,E)$ be the undirected counterpart of a DAG $N=(V,E)$. The metric dimension of~$G$, introduced by \cite{Slater1975} and \cite{Harary1976}, is defined as follows. See \cite{MetricDimensionSurvey} for a detailed survey on the metric dimension of a graph. 

\begin{definition}
	Let $W=\{w_1,w_2,\dots,w_k\}$ be an ordered subset of nodes in a graph $G=(V,E)$ with $w_t\in V$, $t=1,2,\dots,k$. The metric representation of a node $v$ with respect to $W$ is given by
	\begin{equation}
		\bm{r}(v|W)=\big(dis(v,w_1),dis(v,w_2),\dots,dis(v,w_k)\big). 
	\end{equation}
	where $dis(v,w)$ denotes the distance between $v$ and $w$. 
\end{definition}

\begin{definition}
	A set $W=\{w_1,w_2,\dots,w_k\}\subseteq V$ is a resolving set for $G=(V,E)$ if $\forall u,v\in V$, $u\neq v$, we have $\bm{r}(u|W)\neq \bm{r}(v|W)$. 
\end{definition}

\begin{definition}\cite{Slater1975,Harary1976}
	The metric dimension of a graph $G=(V,E)$, denoted as $dim_M(G)$, is the minimum cardinality $k$ of a resolving set $W=\{w_1,w_2,\dots,w_k\}$ for $G$. A resolving set with cardinality $k=dim_M(G)$ is called a basis for $G$. 
\end{definition}
%

Figure \ref{fig:MetricDimEg1} shows a graph with a resolving set with cardinality~2 and the corresponding metric representation of each node. It is easy to verify that the graph has metric dimension~2, as no single node can be a  resolving set of the graph. 
\begin{figure}[ht!]
	\centering
	\includegraphics[width=0.5\columnwidth]{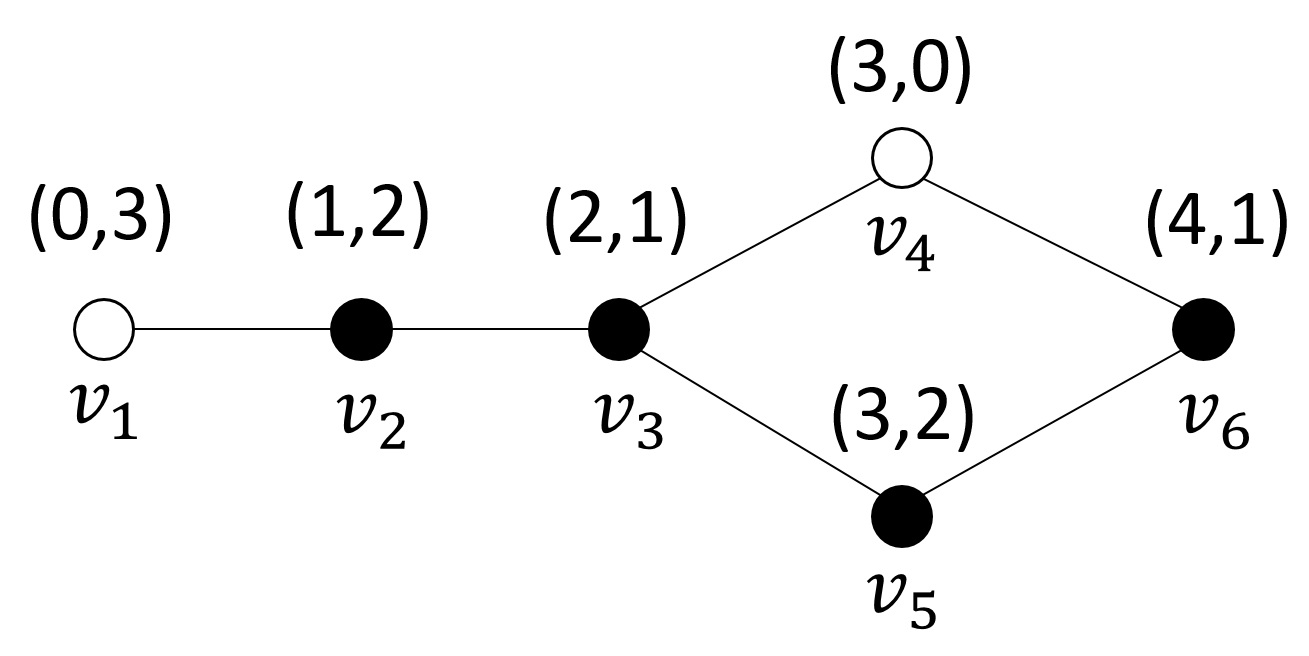}
	\caption{Example of Resolving Set $\bm{\{v_1,v_4\}}$ and Metric Representation.}
	\label{fig:MetricDimEg1}
\end{figure}

A finite graph with $n$ nodes can have metric dimension from 1 to $n-1$. See~\cite{Caceres2009} for characterizations of the metric dimension of some infinite graphs. Finding the metric dimension for general graphs is NP-Complete \cite{Garey1979}. However, we have the following lemmas on the metric dimensions of some special graphs. 
\begin{lemma}\cite[Theorem 2]{MetricDimensionSurvey}\label{lem:metric-dim-tandem}
	The metric dimension of a path is 1.
\end{lemma}
\begin{lemma}\cite[Theorem 3]{MetricDimensionSurvey}\label{lem:metric-dim-complete}
	The metric dimension of a complete graph with size $n$ is $n-1$.
\end{lemma}
\begin{lemma}\cite[Theorem 2.5]{Khuller1996}\label{lem:metric-dim-lattice}
	The metric dimension of a $d$-dimensional grid ($d\ge 2$) is $d$.
\end{lemma}

Next, we adopt the idea of representing nodes by distances to others whereas we focus on how to represent the nodes so that boundedly many nodes can share the same representation as the network size grows to infinite. We propose the following definitions of the extended metric representation and the extended resolving set for a graph. 

\begin{definition}
	Let $\mathcal{W}=\{W_1,W_2,\dots,W_k\}$ be an ordered set of subsets of nodes in a graph $G=(V,E)$ with $W_t\subseteq V$, $t=1,2,\dots,k$. The extended metric representation of a node $v$ with respect to $\mathcal{W}$ is given by
	\begin{equation}
		\bar{\bm{r}}(v|\mathcal{W})=\big(dis(v,W_1),dis(v,W_2),\dots,dis(v,W_k)\big), 
	\end{equation}
	where $dis(v,W_t)$ is the shortest distance between $v$ and any node in $W_t$, $t=1,2,\dots,k$. 
\end{definition}

\begin{definition}
	A set $\mathcal{W}=\{W_1,W_2,\dots,W_k\}$ of subsets of $V$ is a $\Lambda$-extended resolving set for $G=(V,E)$, if $\forall v\in V$, the number of nodes $u\in V$ with $\bar{\bm{r}}(u|\mathcal{W})=\bar{\bm{r}}(v|\mathcal{W})$ is bounded above by a constant $\Lambda>0$. 
\end{definition}
Based on the extended metric representation and the extended resolving set, we propose the following concept of extended metric dimension. 				
\begin{definition}
	Consider an infinite sequence of FJQN/Bs $
	\mathcal{N}=\{N_i\}_{i=1}^\infty$. The \underline{extended metric dimension} of $\mathcal{N}$, denoted as $dim_{EM}(\mathcal{N})$, is the minimum integer~$k$ such that $\forall i\in \mathbb{Z}^+$, the undirected counterpart $G_i=(V_i,E_i)$ of $N_i$ has a $\Lambda$-extended resolving set $\mathcal{W}_i$ with cardinality $\le k$, where $\Lambda>0$ is a constant independent of $i$. 
\end{definition}

In words, the extended metric dimension of a sequence of FJQN/Bs is defined by the minimum integer $k$ such that every network within the sequence has a set of subsets of nodes as a coordinate system that identifies all nodes in the network up to a constant level. This characterizes the dimension of the coordinate system to embed the whole network sequence in a way that constantly many nodes can be mapped to the same position. One can interpret the extended metric dimension as the least number of coordinates needed to describe the network viewed far away as it expands. Note that the extended metric dimension does not rely on the direction of arcs in the networks. Also, note that the definition of the extended metric dimension does not require $\limsup_{i\rightarrow\infty}D_i<\infty$. However, if $\limsup_{i\rightarrow\infty}D_i=\infty$, then we can show that $dim_{EM}(\mathcal{N})=\infty$ as we will need a coordinate system of infinite size to describe a node with infinitely many neighbors. Hence we focus on the cases where $\limsup_{i\rightarrow\infty} D_i\!<\!\infty$ and $\limsup_{i\rightarrow\infty} \Delta_i\!=\!\infty$. 

Regarding with the difference between the metric dimension and the extended metric dimension, we first note that $dim_{EM}(\mathcal{N})\!\le\!\sup_i dim_M(G_i)$ as the extended metric dimension is a generalization of the metric dimension. 
In some cases, the inequality holds tight (see e.g. Remarks~\ref{rmk:tandem-dim} and~\ref{rmk:lattice-dim}). However, in many other cases, the gap exists and can be small, or large, or even infinitely large as shown in the following examples. This observation partially reveals the reason why it is the extended metric dimension but not the metric dimension that determines the throughput scalability of FJQN/Bs. 

\begin{remark}\label{rmk:extm-dim-cycle}
	Consider an example where $N_i$ is a cycle network shown in Figure \ref{fig:cycle-scaling}. Suppose $N_i$ has~$i$ nodes on the top path and $i$ nodes on the bottom path. The metric dimension of the underlying $G_i$ is~$2$, as any pair of adjacent nodes forms a resolving set and any single node cannot resolve the graph. However, the extended metric dimension of $\mathcal{N}=\{N_i\}_{i=1}^\infty$ is $1$, as any single node is a $2$-extended resolving set with cardinality $1$ for $G_i$. 
\end{remark}
\begin{figure}[ht!]
	\centering
	\includegraphics[width=0.35\columnwidth]{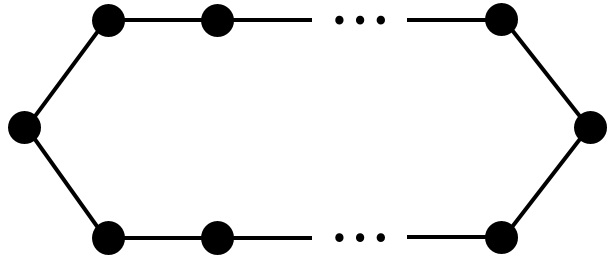}
	\caption{Cycle Network.}
	\label{fig:cycle-scaling}
\end{figure}
\begin{remark}\label{rmk:extm-dim-complete-tandem}
	Consider an example where $N_i$ is a complete-graph $C$ of fixed size $|C|\in \mathbb{Z}^+$ attached by a tandem network of size $i$ shown in Figure \ref{fig:clique-tandem}. The metric dimension of $G_i$ is $|C|-1$ by Lemma~\ref{lem:metric-dim-complete}. However, the extended metric dimension of $\mathcal{N}=\{N_i\}_{i=1}^\infty$ is $1$ by choosing $\mathcal{W}_i=\{C\}$ which forms a $|C|$-extended resolving set with cardinality $1$ for $G_i$, $\forall i\in \mathbb{Z} ^+$. In this case, the gap between the metric dimension and the extended metric dimension can be any large constant depending on $|C|$. 
	Also, note that the scaling dimension of $\mathcal{N}$ is 1 which is determined by the scaling tandem part instead of the fixed-sized complete-graph part. Moreover, if $|C|$ also scales to infinity as $i\rightarrow\infty$, then we have $dim_{EM}(\mathcal{N})=\infty$ as the network degree is not bounded above. 
\end{remark}
\begin{figure}[ht!]
	\centering
	\includegraphics[width=0.55\columnwidth]{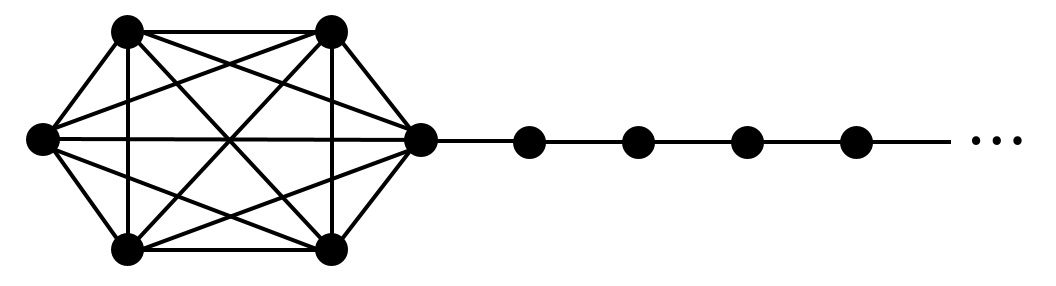}
	\caption{Complete-graph + Tandem Network.}
	\label{fig:clique-tandem}
\end{figure}
\begin{remark}\label{rmk:extm-dim-ladder}
	Consider an example where $N_i$ is a ladder network with $i$ rungs as shown in Figure \ref{fig:ladder}. The metric dimension of the underlying $G_i$ is $i$ (except when $i\!=\!2$ the metric dimension is $3$), as the two nodes on any rung have identical distances to any other node.  However, the extended metric dimension of $\mathcal{N}=\{N_i\}_{i=1}^\infty$ is $1$, as the leftmost rung is a $2$-extended resolving set with cardinality $1$ for $G_i$. In this case, the gap between the metric dimension and the extended metric dimension goes to infinity as $i\rightarrow\infty$. 
\end{remark}
\begin{figure}[ht!]
	\centering
	\includegraphics[width=0.5\columnwidth]{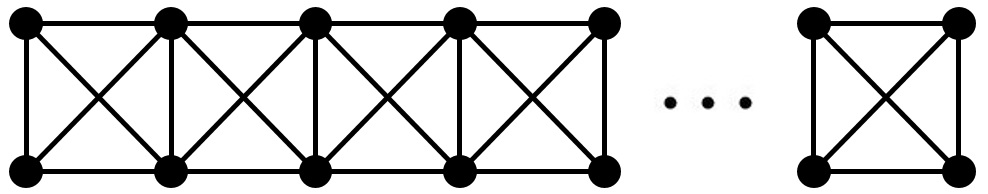}
	\caption{Ladder Network.}
	\label{fig:ladder}
\end{figure}

\subsection{Relationship Between Dimensions}\label{sec:relation-dims}
In this section, we explore the relationship between the scaling dimension and the extended metric dimension. We first show that the scaling dimension is bounded above by the extended metric dimension. Then we conjecture that the extended metric dimension is bounded above by the ceiling of the scaling dimension, which, if true, implies that the gap is less than 1. Finally, we present examples to show the cases when the two dimensions coincide and when there is a gap. 

\begin{lemma}\label{lem:relation-S-le-EM}
	Consider an infinite sequence of FJQN/Bs under Condition~\eqref{eqn:level-degree-bounded}. We have
	\begin{equation}
		dim_S(\mathcal{N})\le dim_{EM}(\mathcal{N}).
	\end{equation}
\end{lemma}
\smallskip

\begin{proof}
	Suppose $dim_{EM}(\mathcal{N})=k$. By definition, there must exists a constant $\Lambda>0$ such that $\forall i\in \mathbb{Z}^+$, the graph $G_i=(V_i,E_i)$ has a $\Lambda$-extended resolving set  $\mathcal{W}_i=\{W_1^{(i)},W_2^{(i)},\dots,W_{k_i}^{(i)}\}$ with cardinality $k_i\le k$. 
	
	
	For all $i\in \mathbb{Z}^+$, consider any node $v^*\in V_i$. For all the nodes $v\in V_i$ that are within $n$ distance away from $v^*$, we have $|dis(v,W_t^{(i)})- dis(v^*,W_t^{(i)})|\!\le\! n$, $\forall t\!=\!1,\dots,k_i$. Thus, the number of distinct $\bm{r}(v|\mathcal{W})$ for all these nodes is bounded above by $(2n+1)^{k_i}$. By definition of the extended resolving set, there are at most $\Lambda (2n+1)^{k_i}$ many nodes in $V_i$ that are within $n$ distance away from any node~$v^*$. Thus, for any subsequence of connected subnetworks $\mathcal{\bar{N}}=\{\bar{N}_{i_n}\}_{n=1}^\infty$, where $\bar{N}_{i_n}=(\bar{V}_{i_n},\bar{E}_{i_n})$, we have $\forall i_n$
	\begin{equation}
		|\bar{V}_{i_n}|\le \Lambda (2\Delta(\bar{N}_{i_n})+1)^{k_i}\le \Lambda (2\Delta(\bar{N}_{i_n})+1)^{k}.
	\end{equation}
	Thus, for any $\mathcal{\bar{N}}=\{\bar{N}_{i_n}\}_{n=1}^\infty$ with $\Delta(\bar{N}_{i_n})\rightarrow\infty$ as $n\rightarrow\infty$, we have
	\begin{equation}
		\limsup_{n\rightarrow\infty}\frac{\log |\bar{V}_{i_n}|}{\log \Delta(\bar{N}_{i_n})}\le \limsup_{n\rightarrow\infty}\frac{\log \left(\Lambda\cdot (2\Delta(\bar{N}_{i_n})+1)^k\right)}{\log \Delta(\bar{N}_{i_n})}=k.
	\end{equation}
	Consequently,
	\begin{equation}
		dim_S(\mathcal{N})=\sup_{\mathcal{\bar{N}}}\left\{
		\limsup_{n\rightarrow\infty}\frac{\log |\bar{V}_{i_n}|}{\log \Delta(\bar{N}_{i_n})}
		\right\}
		\le k= dim_{EM}(\mathcal{N}).
	\end{equation}
\end{proof}\bigskip

We also conjecture that the extended metric dimension is bounded above by the ceiling of the scaling dimension. 
\begin{conjecture}\label{conj:relation-EM-le-S}
	Consider an infinite sequence of FJQN/Bs $
	\mathcal{N}\!=\!\{N_i\}_{i=1}^\infty$ under Condition~\eqref{eqn:level-degree-bounded}. We have 
	\begin{equation}
		dim_{EM}(\mathcal{N})\le \lceil dim_S(\mathcal{N})\rceil.
	\end{equation}
\end{conjecture}

	To verify Conjecture \ref{conj:relation-EM-le-S}, we present the following remarks. 
	The remarks show that in common cases where the scaling dimension is an integer, the two dimensions coincide, as implied by Lemma~\ref{lem:relation-S-le-EM} and Conjecture~\ref{conj:relation-EM-le-S}. The last remark on a Sierpinski triangle shows that in the cases of fractals with non-integer scaling dimension, a non-trivial gap between the two dimensions exists. Moreover, we note that even in the Sierpinski triangle case, Conjecture \ref{conj:relation-EM-le-S} still holds. 
\begin{remark}\label{rmk:tandem-dim}
	Consider a sequence $\mathcal{N}$ that converges to an infinite tandem network (Figure \ref{fig:tandem}). Let $N_i$ be the tandem network with a source and $i$ downstream nodes. 
	In this case, 
	\begin{equation}
		dim_S(\mathcal{N})=dim_{EM}(\mathcal{N})=1.
	\end{equation}
	To see this, note that $dim_S(\mathcal{N})\!=\!1$ as discussed in Remark~\ref{rmk:scal-dim-tandem}. By Lemma \ref{lem:metric-dim-tandem}, graph $G_i$ has metric dimension $1$, $\forall i\!\in\! \mathbb{Z}^+$. Thus, $1\!= \!dim_S(\mathcal{N})\!\le\! dim_{EM}(\mathcal{N})\!\le\! \sup_i dim_{M}(G_i)\!=\!1$. The result also extends to other tandem-alike networks such as series-parallel networks (Table \ref{tbl:summary} (b)), tandem-component networks (Table \ref{tbl:summary}(c)), and ladder networks (Table \ref{tbl:summary}(d)). The extended metric dimension in all these cases remains to be one, although the metric dimension of the underlying graph may not be one. 
\end{remark}

\begin{remark}\label{rmk:lattice-dim}
	Consider a sequence $\mathcal{N}$ that converges to a $d$-dimensional lattice (see e.g. Figure \ref{fig:lattice} for a 2-D lattice). Let $N_i$ be the $d$-dimensional lattice network with $i$ nodes on each side. 
	In this case, 
	\begin{equation}
		dim_S(\mathcal{N})=dim_{EM}(\mathcal{N})=d.
	\end{equation}
	To see this, note that $dim_S(\mathcal{N})\!=\!d$ as discussed in Remark~\ref{rmk:scal-dim-lattice}. By Lemma \ref{lem:metric-dim-lattice}, the underlying graph $G_i$ has metric dimension $d$, $\forall i\!\in\! \mathbb{Z}^+$. Thus, $d\!= \!dim_S(\mathcal{N})\!\le\! dim_{EM}(\mathcal{N})\!\le\! \sup_i dim_{M}(G_i)\!=\!d$. 
\end{remark}

\begin{remark}\label{rmk:hexagon-dim}
	Consider a sequence $\mathcal{N}$ that converges to an infinite hexagon network as shown in Figure \ref{fig:hexagon}. Let $N_i$ be the network with $i$ hexagons on each side. 
	In this case, 
	\begin{equation}
		dim_S(\mathcal{N})=dim_{EM}(\mathcal{N})=2.
	\end{equation}
	To see this, note that $dim_S(\mathcal{N})\!=\!2$ as discussed in Remark~\ref{rmk:scal-dim-hexagon}. 
	On the other hand, we can let $W^{(i)}_1$ be the set of hexagons on one side of $N_i$ and let $W^{(i)}_2$ be the set of hexagons on an adjacent side of $N_i$. The set $\mathcal{W}_i=\{W^{(i)}_1,W^{(i)}_2\}$ resolves each hexagon and hence resolves the entire graph up to a constant level. Thus, $dim_{EM}(\mathcal{N})\!\le\! 2$. In sum, we have $2\!=\! dim_S(\mathcal{N})\!\le\! dim_{EM}(\mathcal{N})\!\le\! 2$. 
\end{remark}

\begin{remark}\label{rmk:pyramid-dim}
	Consider a sequence $\mathcal{N}$ that converges to an infinite tetrahedron pyramid as shown in Figure \ref{fig:tetrahedron}. Let $N_i$ be the network with a vertex and $i$ layers. 
	In this case, 
	\begin{equation}
		dim_S(\mathcal{N})=dim_{EM}(\mathcal{N})=3.
	\end{equation}
	To see this, note that $dim_S(\mathcal{N})\!=\!3$ as discussed in Remark~\ref{rmk:scal-dim-tetrahedron}. 
	On the other hand, we can let $W^{(i)}_1$ be the vertex of the pyramid, let $W^{(i)}_2$ be the set of nodes on one surface of the pyramid that contains the vertex, and let $W^{(i)}_3$ be the set of nodes on another surface of the pyramid that contains the vertex. In this way, every node in $N_i$ has a unique extended metric representation, as distance to $W^{(i)}_1$ represents layer number and distances to $W^{(i)}_2$ and $W^{(i)}_3$ represent the location on a layer. 
	Thus, $dim_{EM}(\mathcal{N})\!\le \!3$. In sum, we have $3\!=\! dim_S(\mathcal{N})\!\le\! dim_{EM}(\mathcal{N})\!\le\! 3$.  
\end{remark}

\begin{remark}\label{rmk:triangle-dim}
	Consider a sequence $\mathcal{N}$ that converges to a Sierpinski triangle as shown in Figure \ref{fig:Sierpinski}. Let $N_i$ be the network with $2^{i-1}$ edges on the side of the largest triangle. In this case, the extended metric dimension equals the celling of the scaling dimension, i.e.
	\begin{equation}
		dim_{EM}(\mathcal{N})=\lceil dim_{S}(\mathcal{N})\rceil=\lceil \log_2 3\rceil=2.
	\end{equation}
	To see this, note that no single subset of nodes in $N_i$ can resolve the graph up to any constant level. But we can let $W^{(i)}_1$ be the set of nodes on one side of the largest triangle in $N_i$ and let $W^{(i)}_2$ be the set nodes on another side of the largest triangle so that $\mathcal{W}_i=\{W^{(i)}_1,W^{(i)}_2\}$ can fully resolve the graph. Hence $dim_{EM}(\mathcal{N})=2$. Meanwhile, $dim_{S}(\mathcal{N})=\log_2 3$ as discussed in Remark~\ref{rmk:scal-dim-triangle}. This example verifies our Conjecture \ref{conj:relation-EM-le-S}. 
\end{remark}

\section{Examples}\label{sec:app}
This section further discusses the examples illustrated in Table~1. We illustrate how to use the relations between network dimensions and service time tails to identify throughput scalability. Recall that service times are i.i.d. regularly varying with index $\alpha\in \mathbb{R}^+$. 

\bigskip

\noindent{\bf Tandem/Tandem-alike Network}:
Consider a tandem FJQN/B as shown in Table~\ref{tbl:summary}(a). 
Let $N_i$ be the tandem network with a source and $i$ downstream nodes.

For this system,  $dim_{S}(\mathcal{N})=dim_{EM}(\mathcal{N})=1$ as discussed in Remark \ref{rmk:tandem-dim}. Based on Theorem \ref{thm:main-result}, the network throughput is scalable if $\alpha>2$ and only if $\alpha \ge 2$. 

Our results also extend to other tandem-alike networks, e.g. a series-parallel network with bounded network degree (Table~\ref{tbl:summary}(b)), 
a tandem-component network with finite-size components (Table~\ref{tbl:summary}(c)), 
and a ladder network (Table~\ref{tbl:summary}(d)). 
For all these networks as they expand, the scaling dimension and the extended metric dimension are equal to~1. Hence the throughput scalability conditions are the same as those in the tandem case. 

\bigskip

\noindent{\bf Lattice Network}:
Consider a 2-D lattice FJQN/B as shown in Table~\ref{tbl:summary}(e). 
Let $N_i$ be the 2-D lattice network with $i\times i$ nodes.

For this system,  $dim_{S}(\mathcal{N})=dim_{EM}(\mathcal{N})=2$ as discussed in Remark \ref{rmk:lattice-dim}. Based on Theorem \ref{thm:main-result}, the network throughput is scalable if $\alpha>3$ and only if $\alpha \ge 3$. Similar to the discussion of tandem and tandem-alike networks, the results here can be extended to other lattice-alike networks. In general, a $d$-dimensional lattice/lattice-alike FJQN/B with growing size is throughput scalable if $\alpha > d+1$ and only if $\alpha \ge d+1$. 

\bigskip

\noindent{\bf Hexagon Network}:
Consider a hexagon FJQN/B as shown in Table~\ref{tbl:summary}(g). 
Let $N_i$ be the hexagon network with $i$ hexagons on each side.

For this system,  $dim_{S}(\mathcal{N})=dim_{EM}(\mathcal{N})=2$ as discussed in Remark \ref{rmk:hexagon-dim}. Based on Theorem \ref{thm:main-result}, the network throughput is scalable if $\alpha>3$ and only if $\alpha \ge 3$. 

\bigskip

\noindent{\bf Tetrahedron Pyramid Network}:
Consider a tetrahedron pyramid FJQN/B with growing layers as shown in Table~\ref{tbl:summary}(h). 
Let $N_i$ be the pyramid network with a vertex and $i$ layers. 

For this system,  $dim_{S}(\mathcal{N})=dim_{EM}(\mathcal{N})=3$ as discussed in Remark \ref{rmk:pyramid-dim}. Based on Theorem \ref{thm:main-result}, the network throughput is scalable if $\alpha>4$ and only if $\alpha \ge 4$. 

\bigskip

\noindent{\bf Fractal Network}:
Consider a sequence of FJQN/Bs that converges to a Sierpinski triangle as shown in Table~\ref{tbl:summary}(i). 
Let $N_i$ be the network with $2^{i-1}$ edges on the side of the largest triangle. 

For this system, $dim_{S}(\mathcal{N})=\log_2 3$ and  $dim_{EM}(\mathcal{N})=2$ as discussed in Remark~\ref{rmk:triangle-dim}. 
Based on Theorem \ref{thm:main-result}, the network throughput is scalable if $\alpha>3$ and only if $\alpha \ge 1+\log_2 3$. 

\bigskip

\noindent{\bf Binary Tree Network}:
Consider a binary tree FJQN/B with growing leaves as shown in Table~\ref{tbl:summary}(j). 
Let $N_i$ be the binary tree network with a root and $i$ layers.

The network has exponential growth and $dim_{S}(\mathcal{N})=dim_{EM}(\mathcal{N})=\infty$. Based on Theorem \ref{thm:main-result}, the network throughput is not scalable for any $\alpha\in \mathbb{R}^+$. Similar discussion appears in \cite{Chaintreau} where light-tailed service time distribution is shown necessary for throughput scalability. In general, we have the following corollary. 
\begin{corollary}
	Consider an infinite sequence of FJQN/Bs $
	\mathcal{N}=\{N_i\}_{i=1}^\infty$ with $\limsup_{i\rightarrow\infty}|V_i|=\infty$. Suppose service times are i.i.d. regularly varying with index $\alpha\in \mathbb{R}^+$. 
	If $N_i$ grows exponentially fast, i.e. there exist constants $C_1,C_2>0$ such that 
	$|V_i|\ge C_1\cdot e^{C_2\Delta_i}$, $\forall i\in \mathbb{Z}^+$, then the sequence $\mathcal{N}$ is not throughput scalable. 
\end{corollary}

\section{Proof of Theorem 1}\label{sec:pf-main}
Consider an infinite sequence of FJQN/Bs $
\mathcal{N}=\{N_i\}_{i=1}^\infty$, where $N_i=(V_i,E_i)$ is associated with an underlying undirected graph~$G_i$, network degree $D_i$, minimum level $L^*_i$, and diameter $\Delta_i$. We first propose the following theorems. See Appendices~\ref{apd:pf-nec} and~\ref{apd:pf-suf} for the detailed proofs. 

\begin{theorem}\label{thm:nec}
	Consider an infinite sequence of FJQN/Bs $
	\mathcal{N}=\{N_i\}_{i=1}^\infty$ with $\limsup_{i\rightarrow\infty}|V_i|=\infty$. Suppose service times are i.i.d. regularly varying with index $\alpha\!>\!1$. Under Condition~\eqref{eqn:level-degree-bounded}, the sequence $\mathcal{N}$ is NOT throughput scalable if $dim_S(\mathcal{N})>\alpha-1$.
\end{theorem}

\begin{theorem}\label{thm:suf}
	Consider an infinite sequence of FJQN/Bs $\mathcal{N}=\{N_i\}_{i=1}^\infty$ with $\limsup_{i\rightarrow\infty}|V_i|=\infty$. Suppose service times are i.i.d. regularly varying with index $\alpha\!>\!1$. Under Condition~\eqref{eqn:level-degree-bounded}, the sequence $\mathcal{N}$ is throughput scalable if $\exists K\in \mathbb{Z}^+$, $K\ge 1$ such that \eqref{eqn:FsigmaK} holds and  $dim_{EM}(\mathcal{N})\le K-1$. 
\end{theorem}

Theorems \ref{thm:nec} and \ref{thm:suf} allow us to complete the proof of the main result. The necessary part directly follows from Theorem \ref{thm:nec}. Note that when $\limsup_{i\rightarrow\infty}|V_i|=\infty$ and Condition~\eqref{eqn:level-degree-bounded} holds, we must have $\limsup_{i\rightarrow\infty}\Delta_i=\infty$ and hence the network scaling dimension is well defined. The sufficient part is by having
\begin{equation}
	K=dim_{EM}(\mathcal{N})+1<\alpha.
\end{equation}
Thus, there exists $0<\epsilon<\alpha-K$ such that~$\mathbb{E}\left[\sigma^{K+\epsilon}\right]<~\!\infty$ which is a condition stronger than \eqref{eqn:FsigmaK}. Then applying Theorem \ref{thm:suf} yields the result. \qed

\bigskip

Lemma \ref{lem:relation-S-le-EM} and Conjecture \ref{conj:relation-EM-le-S} imply that when the scaling dimension is an integer, it equals the extended metric dimension. In such case, we have the following corollary. 

\begin{corollary}
	Consider an infinite sequence of FJQN/Bs $\mathcal{N}=\{N_i\}_{i=1}^\infty$ with $\limsup_{i\rightarrow\infty}|V_i|=\infty$. Suppose service times are i.i.d. regularly varying with index $\alpha$. Suppose Condition~\eqref{eqn:level-degree-bounded} holds and $dim_S(\mathcal{N})=dim_{EM}(\mathcal{N})=K-1$, where $K\ge 2$, $K\in \mathbb{Z}^+$. The sequence $\mathcal{N}$ is throughput scalable if $\mathbb{E}\left[\sigma^{K+\epsilon}\right]<\infty$ and only if $\mathbb{E}\left[\sigma^{K}\right]<\infty$. 
\end{corollary}

\section{Conclusion}\label{sec:conclusion}
This paper investigates throughput scalability of fork-join queueing networks with blocking under heavy-tailed service times. In particular, we focus on cases where service times are regularly varying with index $\alpha$. 
We introduce two topological concepts for generally structured FJQN/Bs: scaling dimension and extended metric dimension. 
We show that a sequence of FJQN/Bs is throughput scalable if its extended metric dimension $<\alpha-1$ and only if its scaling dimension $\le \alpha-1$. The results apply to a list of FJQN/Bs including tandem, lattice, hexagon, and tetrahedron pyramid networks, where the two dimensions coincide and the proposed conditions are almost tight. Even for fractals where the two dimensions don't coincide, we conjecture the gap between the necessary condition and the sufficient condition is less than one. Our analysis is based on last-passage percolation, extreme value theory, and lattice animal argument. The results can be useful for designing large-scale parallel and distributed processing systems in heavy-tailed service time environment as well as for analysis of other scaling networks or fractals such as social networks, electrical grid, Internet of Things, etc. 

In this paper, we conjecture that the extended metric dimension is upper bounded by the ceiling of the scaling dimension. Future research could focus on exploring the conjecture so as to close the gap for fractals with non-integer Hausdorff dimensions. A second research direction is to investigate the benefit of job replication strategies for system performance. Job replication strategies are applied in~\cite{Wang2015} in parallel computing systems with heavy-tailed execution times to reduce latency. It is still open to devise the optimal replication strategy for general parallel and distributed processing systems and to understand the role of job replication in guaranteeing throughput scalability. A third direction is to generalize our results to scenarios where resource capabilities, such as storage and processing speed, are also improving as the network scales in size. These improvements should intuitively reduce synchronization burden of parallel and distributed processing systems and mitigate throughput degradation. Lastly, it is interesting to generalize our analyses from the basic FCFS queueing discipline to other disciplines, such as processor sharing and priority rules, so as to design scalable computing systems with processor sharing virtual machines and user-specified quality-of-service targets.

\appendix
\section{Preliminaries on Precedence \\Graph}\label{apd:precedence}
Given a FJQN/B $N=(V, E)$,  the corresponding precedence graph $\mathcal{G}=(\mathcal{V},\mathcal{E})$ 
has been defined in section \ref{sec:precedence}. The following lemmas will be useful in later proofs. 

Name the arcs in $\mathcal{E}^I,\mathcal{E}^{II},\mathcal{E}^{III}$ (defined in \eqref{eqn:G2})
as Type~I, Type II, Type III arcs, respectively. 
The following lemma provides an upper bound on the number of arcs on a path in the precedence graph. 

\begin{lemma}\label{lem:path-upper-bound}
	Consider a path $\pi_{(m,v)\leadsto(0,v)}$ from $(m,v)$ to $(0,v)$ for any $v\in V$. The number of arcs on $\pi_{(m,v)\leadsto(0,v)}$ is upper bounded by
	\begin{equation}
		\left|\pi_{(m,v)\leadsto(0,v)}\right|\le \max\{L^*(N)+1,b\}\cdot m/b,
	\end{equation}
	and the number of Type III arcs on $\pi_{(m,v)\leadsto(0,v)}$ is upper bounded by
	\begin{equation}
		\left|\pi^{III}_{(m,v)\leadsto(0,v)}\right|\le m/b.
	\end{equation}
\end{lemma}

\begin{proof}
	Introduce the following function
	\begin{equation}
		\phi:(m,v)\mapsto \max\{(c+1)/b,1\}\cdot m+l^*_N(v),
	\end{equation}
	where $c=L^*(N)$ is the minimum level of the FJQN/B network $N=(V, E)$, and $l^*_N(v)$ is the corresponding optimal topological labelling on node $v$. 
	
	For Type I arcs: $(m,v)\rightarrow (m,u)$ with $(u,v)\in E$, we have
	\begin{equation}
		\phi(m,v)-\phi(m,u)=l^*_N(v)-l^*_N(u)\ge 1.
	\end{equation}
	For Type II arcs: $(m,v)\rightarrow (m-1,v)$ with $m\ge 1$, we have
	\begin{equation}
		\phi(m,v)-\phi(m-1,v)=\max\{(c+1)/b,1\}\ge 1.
	\end{equation}
	For Type III arcs: $(m,v)\rightarrow (m-b,u)$ with $(v,u)\in E,m\ge b$, we have
	\begin{eqnarray}
		\phi(m,v)\!-\!\phi(m\!-\!b,u)\!&\!\!\!=\!\!\!&\!\max\{(c\!+\!1)/b,1\}b \!-\!\left[l^*(u)\!-\!l^*(v)\right]\nonumber\\
		\!&\!\!\!\ge\!\!\!&\! \max\{c+1,b\}-c\ge 1.
	\end{eqnarray}
	In summary, we know that $\phi$ decreases by at least one for each arc in $\mathcal{E}$. Thus,
	\begin{eqnarray}
		\left|\pi_{(m,v)\leadsto(0,v)}\right|&\le & \phi(m,v)-\phi(0,v)\nonumber\\
		&=&\max\{(c+1)/b,1\}\cdot m\nonumber\\
		&=&\max\{c+1,b\}\cdot m/b.
	\end{eqnarray}		
	The upper bound on $\left|\pi^{III}_{(m,v)\leadsto(0,v)}\right|$ is trivial, as each Type~III arc goes from $m$ to $m-b$. 
\end{proof}
%
\smallskip 

Let $Wei(\pi^*_{(m,v)\leadsto(m'',v'')})\!=\!\max\{Wei(\pi)\big| \pi\!:\!(m,v)\!\leadsto\!(m'',v'')\}$. We have the following lemma on the super-additive property of the maximum weighted path (see e.g. \cite{Chaintreau,Martin,Martin-LPP}). 
\begin{lemma}\label{lem:super-additive}
	$\forall v,v',v''\in V$, $\forall m,m',m''$ such that paths $(m,v)\leadsto(m',v')$ and $(m',v')\leadsto(m'',v'')$ exist in $\mathcal{G}$, 
	\begin{eqnarray}
		Wei(\pi^*_{(m,v)\leadsto(m'',v'')})&\!\!\!\ge\!\!\!&   Wei(\pi^*_{(m,v)\leadsto(m',v')})\nonumber\\
		&&+  Wei(\pi^*_{(m',v')\leadsto(m'',v'')})\nonumber\\
		&&-S_{m',v'}(N).
	\end{eqnarray}
\end{lemma}

\smallskip
The following lemma is immediate by construction. 
\begin{lemma}\label{lem:exist-path}
	$\forall v,v'\in V$, $\forall m,m'$ such that $m-m'\ge \Delta(N) b$, we can always find a path $\pi\!:\!(m,v)\!\leadsto\!(m',v')$ in $\mathcal{G}$. 
\end{lemma}

\section{Maxima and Extreme Value Theory}
Let $M_n=\max\left\{\sigma_1,\sigma_2,\dots,\sigma_n \right\}$, where $\sigma_1,\sigma_2,\dots,\sigma_n$ are i.i.d. with distribution $F_{\sigma}$. Suppose there exist normalizing constants $a_n>0, b_n$ such that 
\begin{equation}\label{eqn:MDA}
	\mathbb{P}\left[\frac{M_n-b_n}{a_n}\le x\right]=F_{\sigma}(a_nx+b_n)^n\rightarrow H(x)~\text{as}~n\rightarrow \infty.
\end{equation}
When \eqref{eqn:MDA} holds, we say that $F_{\sigma}$ belongs to the Maximum Domain of Attraction (MDA) of $H$.  
By extreme value theory (see e.g. \cite{EVTGalambos,EVTEmbrechts,EVT}), $H(x)$ must fall into one of three distribution classes: Weibull, Gumbel, and Fr{\'e}chet. 
When $F_{\sigma}$ has a short tail or an exponential tail, it belongs to the MDA of $H$ for the first two classes, respectively. 
When $F_{\sigma}$ is regularly varying with index $\alpha>0$, it has a long tail and belongs to the MDA of $H$ for the Fr{\'e}chet class 
\cite{EVT},
where 
\begin{eqnarray}
	\text{Fr{\'e}chet}: & \Phi_{\alpha}(x)=
	\left\{\begin{aligned}
		& 0, && x< 0;\\
		& \text{exp}\{-x^{\alpha}\}, && x\ge 0,\alpha >0.
	\end{aligned}\right.
\end{eqnarray}

\section{Proof of Theorem 2}\label{apd:pf-nec}
Suppose $dim_S(\mathcal{N})>\alpha-1$. We will show the sequence $\mathcal{N}$ is not throughput scalable. The proof consists of two parts: a) constructing an upper bound on liminf of throughput; b) showing the upper bound equals zero by extreme value theory. 

\bigskip
\noindent {\it Part a: constructing an upper bound on liminf of throughput.}

Since $dim_S(\mathcal{N})>\alpha-1$, there must exist a constant $K$ such that $dim_S(\mathcal{N})>K-1>\alpha-1$. By Definition~\ref{def:scaling-dim},
there must exist a subsequence of subnetworks $
\mathcal{\bar{N}}=\{\bar{N}_{i_j}\}_{j=1}^\infty$ with $\bar{V}_{i_j}\subseteq V_{i_j}$, $\bar{E}_{i_j}\subseteq E_{i_j}$, $\Delta(\bar{N}_{i_j})\rightarrow\infty$ as $j\rightarrow\infty$, and 
\begin{equation}
	\limsup_{j\rightarrow\infty}\frac{\log |\bar{V}_{i_j}|}{\log \Delta(\bar{N}_{i_j})}>K-1. 
\end{equation}
This implies there further exists a subsequence of $
\mathcal{\bar{N}}$, denoted as $
\mathcal{\bar{N}}'=\{\bar{N}_{i_{j_n}}\}_{n=1}^\infty$, with $\Delta(\bar{N}_{i_{j_n}})\rightarrow\infty$ as $n\rightarrow\infty$ and
$\lim_{n\rightarrow\infty}\frac{\log |\bar{V}_{i_{j_n}}|}{\log \Delta(\bar{N}_{i_{j_n}})}>K-1.$ 
For simplicity, let $\bar{N}'_{n}= \bar{N}_{i_{j_n}}$,  $\bar{V}'_{n}= \bar{V}_{i_{j_n}}$, and $\bar{E}'_{n}=\bar{E}_{i_{j_n}}$, for $n\in\mathbb{Z}^+$. Then $\mathcal{\bar{N}}'=\{\bar{N}'_{n}\}_{n=1}^\infty$ is an infinite sequence of connected FJQN/Bs with  $\Delta(\bar{N}'_{n})\rightarrow\infty$ as $n\rightarrow\infty$, and 
\begin{equation}\label{eqn:hausdorff-ge-K-1}
	\lim_{n\rightarrow\infty}\frac{\log |\bar{V}'_{n}|}{\log \Delta(\bar{N}'_{n})}>K-1. 
\end{equation}
By the monotonicity on throughput with respect to network inclusion (see \cite{Bacelli}) and the construction of subsequences, 
\begin{equation}
	\liminf_{i\rightarrow \infty} \theta(N_i)\!\le\! \liminf_{j\rightarrow \infty} \theta(N_{i_j})\!\le\! 
	\liminf_{j\rightarrow \infty} \theta(\bar{N}_{i_j})\!\le\! 
	\liminf_{n\rightarrow \infty} \theta(\bar{N}'_{n}).
\end{equation}

Next, we will construct an upper bound on $\liminf_{n\rightarrow \infty} \theta(\bar{N}'_{n})$. 
Denote $\Delta(\bar{N}'_{n})\!=\!\Delta_n$. For large $m$ as a multiple of $3\Delta_nb$, by dividing $[0,m]$ into equal intervals of length  $3\Delta_nb$, we can partition the precedence graph $\mathcal{P}_n$ of network $\bar{N}'_{n}$ into $\frac{m}{3\Delta_nb}$ layers. Let $\pi^*_j$ denote the maximum weighted path in layer $j$ from  $(m-3j\Delta_nb,v)$ to $(m-3(j+1)\Delta_nb,v)$ where $j=0,1,\dots,\frac{m}{3\Delta_nb}-1$. By the super-additive property (see Lemma \ref{lem:super-additive}), the weight of the maximum weighted path from $(m,v)$ to $(0,v)$ is bounded below by the sum of $Wei(\pi^*_j)-S_{m-3j\Delta_nb,v}$ over all $j$. Essentially, this lower bound comes from adding a constraint that the path contains nodes $(m-3j\Delta_nb,v)$ for all~$j$. Together with Lemma~\ref{lem:max-weight}, we have
\begin{equation}
	T_{m,v}(\bar{N}'_{n})\ge\!\!\! \sum_{j=0}^{\frac{m}{3\Delta_nb}-1}Wei(\pi^*_j)-S_{m-3j\Delta_nb,v}(\bar{N}'_{n}). 
\end{equation} 
Within each layer $j$, by Lemma \ref{lem:exist-path}, we can always find a path $(m-3j\Delta_nb,v)\leadsto (m',v')\leadsto (m-3(j+1)\Delta_nb,v)$ in the precedence graph $\mathcal{G}_n$ for any $(m',v')$ satisfying
\begin{eqnarray}
	m'&\!\!\!\in\!\!\!&\left[m-3j\Delta_nb-2\Delta_nb,m-3j\Delta_nb-\Delta_nb-1\right]\\
	v'&\!\!\!\in\!\!\!&\bar{V}'_{n}
\end{eqnarray} 
Therefore, $Wei(\pi^*_j)-S_{m-3j\Delta_nb,v}(\bar{N}'_{n})$ must be larger than the maximum weight of $(m',v')$ among all above possible choices. Using Lemma~\ref{lem:max-weight} and the above discussion on the maximum weighted path, we have
\begin{equation}\label{eqn:Tmv-lowerbound}
	T_{m,v}(\bar{N}'_{n})\ge\!\!\! \sum_{j=0}^{\frac{m}{3\Delta_nb}-1}\max\left\{S_{m',v'}\left|
	\begin{aligned}
		m'&\ge m- 3j\Delta_nb-2\Delta_nb\\
		m'&< m- 3j\Delta_nb-\Delta_nb\\
		v'&\in \bar{V}'_{n}
	\end{aligned}
	\right.
	\right\}.
\end{equation}
Since the service times are i.i.d., we have
\begin{equation}
	\mathbb{E}\left[\text{RHS in \eqref{eqn:Tmv-lowerbound}}\right]=
	\frac{m}{3\Delta_nb}\mathbb{E}\left[\max\{S_{m',v'}|0\!\le\! m'\!<\!\Delta_nb, v'\!\in\! \bar{V}'_n\}\right].
\end{equation}
Thus,
\begin{equation}
	\mathbb{E}\left[T_{m,v}(\bar{N}'_{n})\right]\ge
	\frac{m}{3\Delta_nb}\mathbb{E}\left[\max\{S_{m',v'}|0\!\le\! m'\!<\!\Delta_nb, v'\!\in\! \bar{V}'_n\}\right].
\end{equation}
Thus, by throughput definition \eqref{eqn:throughput-def}, 
\begin{equation}\label{eqn:App-B-1}
	\theta(\bar{N}'_{n})\le\left(\frac{\mathbb{E}\left[\max\{S_{m',v'}|0\le m'< \Delta_nb, v'\in \bar{V}'_n\}\right]}{3\Delta_nb}\right)^{-1}
	.
\end{equation}
Equation \eqref{eqn:hausdorff-ge-K-1} implies there exist constants $c>0, n_0\in \mathbb{Z}^+$ such that 
\begin{equation}
	|\bar{V}'_n|\ge c \cdot \Delta_n^{K-1}, ~~~\forall n\ge n_0.
\end{equation} 
Thus, for all $n\ge n_0$, the total number of different choices in the max term in \eqref{eqn:App-B-1} is bounded below by $c\Delta_n^{K-1}\Delta_nb$. Consider i.i.d. random variables $\{\sigma_h\}_{h\ge 1}$ following distribution~$F_\sigma$. For all $n\ge n_0$, we have
\begin{equation}
	\mathbb{E}\!\left[\max\{S_{m',v'}|0\!\le \!m'\!<\! \Delta_nb, v'\!\in\! \bar{V}'_n\}\right]
	\ge\mathbb{E}\left[M_{g_n}\right].
\end{equation}
where $g_n:= c \Delta_n^{K}b$ and  $M_{h}:=\max\left\{\sigma_1,\sigma_2,\dots, \sigma_{h}\right\}$. 
Thus,
\begin{equation}
	\liminf_{i\rightarrow \infty} \theta(N_i)\le \liminf_{n\rightarrow \infty} \theta(\bar{N}'_n)\le \left(\limsup_{n\rightarrow \infty}\frac{\mathbb{E}\left[M_{g_n}\right]}{3\Delta_nb}\right)^{-1}\!\!\!.
\end{equation}

\bigskip

\noindent {\it Part b: showing the upper bound equals zero by extreme value theory.}

By construction, $\Delta_n\rightarrow\infty$ as $n\rightarrow\infty$ and hence $g_n\rightarrow\infty$ as $n\rightarrow\infty$. As $F_{\sigma}$ is regularly varying with index $\alpha$, by extreme value theory \cite[Theorem 3.3.7]{EVTEmbrechts}, we know that $F_{\sigma}$ belongs to the Maximum Domain of Attraction of Fr{\'e}chet distribution~$\Phi_\alpha$ and
\begin{equation}
	Y_n:=\frac{\max\left\{\sigma_1,\sigma_2,\dots, \sigma_{g_n}\right\}}{g_n^{1/\alpha}L_1(g_n)}\xrightarrow{d}\Phi_\alpha,
\end{equation}
where $L_1$ is some slowly varying function. 
By the property of convergence in distribution (see e.g. \cite[Theorem 25.11]{ProbMeasureBP}), we have
\begin{equation}
	\liminf_{n\rightarrow \infty} \mathbb{E}\left[Y_n\right]\ge \mathbb{E}\left[\Phi_\alpha\right]>0.
\end{equation}
This implies that $\forall \epsilon>0$, there exists $n_\epsilon>0$ such that $\forall n\ge n_\epsilon$ we have 
\begin{equation}
	\mathbb{E}\left[Y_n\right]\ge \liminf_{n\rightarrow \infty} \mathbb{E}\left[Y_n\right]-\epsilon\ge \mathbb{E}\left[\Phi_\alpha\right]-\epsilon.
\end{equation}
Let $0<\epsilon<\mathbb{E}\left[\Phi_\alpha\right]$. We have, $\forall n\ge n_\epsilon$ 
\begin{equation}\label{eqn:nec-heavy-pf-liminfEYn}
	\mathbb{E}\left[Y_n\right]\ge \mathbb{E}\left[\Phi_\alpha\right]-\epsilon>0.
\end{equation}
Now consider
\begin{equation}\label{eqn:nec-heavy-pf-limsupMax}
	\limsup_{n\rightarrow \infty}\frac{\mathbb{E}\left[M_{g_n}\right]}{3\Delta_nb}=
	\frac{(cb)^{1/K}}{3b}\limsup_{n\rightarrow \infty}\left(\mathbb{E}\left[Y_n\right]\cdot L_1(g_n)\cdot g_n^{1/\alpha-1/K}\right).
\end{equation}
Note that $\lim_{n\rightarrow \infty}g_n^{1/\alpha-1/K}=\infty$ as $\alpha<K$. By the basic property of slowly varying function (see e,g. \cite[Remark 1.2.3]{RegularVariation-Mikosch}), we have
\begin{equation}
	\lim_{n\rightarrow \infty}\left(L_1(g_n)\cdot g_n^{1/\alpha-1/K}\right)=\infty.
\end{equation}
Applying \eqref{eqn:nec-heavy-pf-liminfEYn} yields
\begin{eqnarray}
	&&\limsup_{n\rightarrow \infty}\left(\mathbb{E}\left[Y_n\right]\cdot L_1(g_n)\cdot g_n^{1/\alpha-1/K}\right)\nonumber\\
	&\!\!\!\ge\!\!\! &
	\limsup_{n\rightarrow \infty}\left(\left(\mathbb{E}\left[\Phi_\alpha\right]-\epsilon\right)\cdot L_1(g_n)\cdot g_n^{1/\alpha-1/K}\right)\nonumber\\
	&\!\!\!= \!\!\!&\left(\mathbb{E}\left[\Phi_\alpha\right]-\epsilon\right)
	\limsup_{n\rightarrow \infty}\left( L_1(g_n)\cdot g_n^{1/\alpha-1/K}\right)\nonumber\\
	&\!\!\!=\!\!\!&\infty.
\end{eqnarray}
Thus,
\begin{equation}
	\limsup_{n\rightarrow \infty}\frac{\mathbb{E}\left[M_{g_n}\right]}{3\Delta_nb}=\infty.
\end{equation}
Consequently,
\begin{equation}
	\liminf_{i\rightarrow \infty} \theta(N_i)\le \liminf_{n\rightarrow \infty} \theta(\bar{N}'_n)\le\left(\limsup_{n\rightarrow \infty}\frac{\mathbb{E}\left[M_{g_n}\right]}{3\Delta_nb}\right)^{-1}\!=0.
\end{equation}

\section{Proof of Theorem 3}\label{apd:pf-suf}
Before showing Theorem \ref{thm:suf}, we need the following lemma which says that if the metric dimension of all $G_i$'s in the sequence is no larger than $K-1$, then we can embed all precedence graphs $\mathcal{P}_i$'s of $N_i$'s into a $K$-dimensional lattice, and Condition~\eqref{eqn:FsigmaK} ensures that the sequence is throughput scalable. 
\begin{lemma}\label{lem:scal-suf-metric-d}
	Consider an infinite sequence of FJQN/Bs $\mathcal{N}=\{N_i\}_{i=1}^\infty$ under i.i.d. regularly varying service times with index $\alpha\!>\!1$. Under Condition~\eqref{eqn:level-degree-bounded}, the sequence $\mathcal{N}$ is throughput scalable if $\exists K\in \mathbb{Z}^+$, $K\ge 1$ such that 
	\eqref{eqn:FsigmaK} holds 
	and 
	$dim_M(G_i)\le K-1$, $\forall i\in \mathbb{Z}^+$.   
\end{lemma}
\begin{proof} {(\bf for Lemma~\ref{lem:scal-suf-metric-d}).} The proof consists of two parts: a) constructing a lower bound on liminf of throughput; b) showing the lower bound is strictly positive by mapping networks onto lattices. 
	
	\bigskip
	\noindent {\it Part a: constructing a lower bound on liminf of throughput.}
	
	Let  $\tilde{\Delta}_{i+1}=\max\{\Delta_{i+1}, \tilde{\Delta}_{i}\}$ $\forall i\in \mathbb{Z}^+$, and $\tilde{\Delta}_{1}=\Delta_1$. We have $\tilde{\Delta}_{i}\ge \Delta_i$ and $\tilde{\Delta}_{i+1}\ge \tilde{\Delta}_{i}$,  $\forall i\in \mathbb{Z}^+$. Note that $\{\tilde{\Delta}_i\}_{i=1}^\infty$ is monotone and $\lim_{i\rightarrow\infty} \tilde{\Delta}_i\!=\!\infty$. 
	
	Let $\mathcal{P}_i=(\mathcal{V}_i,
	\mathcal{E}_i)$ be the precedence graph of $N_i$. 
	By Lemma \ref{lem:max-weight}, $T_{m,v}(N_i)$ is given by the maximum weighted path from $(m,v)$ to $(0,v')$ in  $\mathcal{P}_i$ for all $v'\in V_i$. Denote this path by $\Pi^*$. For large $m$ as a multiple of $\tilde{\Delta}_ib$, by dividing $[0,m]$ into equal intervals of length  $\tilde{\Delta}_ib$, we can partition $\mathcal{P}_i$ into $\frac{m}{\tilde{\Delta}_ib}$ layers, where 
	\begin{equation}
		\text{layer $j$:~}[m-(j+1)\tilde{\Delta}_ib,m-j\tilde{\Delta}_ib],j\!=\!0,1,\dots,\frac{m}{\tilde{\Delta}_ib}-1. 
	\end{equation} 
	The chunk of $\Pi^*$ within layer $j$ can always be fully covered by some path $\pi_j$ from the top of the layer to the bottom of the layer, where the weight of $\pi_j$ is bounded above by the weight of the maximum weighted path in layer $j$ as follows. 
	\begin{equation}
		Wei(\pi_j)\le \max_{\pi}\left\{Wei(\pi)\left|
		\begin{aligned}
			\pi & \text{~from~}  (m-j\tilde{\Delta}_ib,v'_j)\\
			& \text{~to~}  (m-(j+1)\tilde{\Delta}_ib,v''_j)\\
			\forall & v'_j,v''_j  \in V_i
		\end{aligned}
		\right.
		\right\}.
	\end{equation} 
	Thus, we can upper bound $T_{m,v}(N_i)$ by the sum weight of maximum weighted paths in each layer, namely,
	\begin{equation}\label{eqn:Tmv-upperbound}
		T_{m,v}(N_i)
		\!\le\!\!\! \sum_{j=0}^{\frac{m}{\tilde{\Delta}_ib}-1}\!\!\!\max_{\pi}\left\{Wei(\pi)\left|
		\begin{aligned}
			\pi & \text{~from~}  (m-j\tilde{\Delta}_ib,v'_j)\\
			& \text{~to~}  (m-(j+1)\tilde{\Delta}_ib,v''_j)\\
			\forall & v',v''  \in V_i
		\end{aligned}
		\right.
		\right\}.
	\end{equation}
	Essentially, this upper bound comes from relaxing the constraint that the path is connected between adjacent layers. Since the service times are i.i.d., we have
	\begin{equation}
	\mathbb{E}\left[\text{RHS in \eqref{eqn:Tmv-upperbound}}\right]=
	\frac{m}{\tilde{\Delta}_ib}\mathbb{E}\left[\max_{\pi}\!\left\{Wei(\pi)\left|
	\begin{aligned}
	\pi & \text{~from~}  (2\tilde{\Delta}_ib,v')\\
	& \text{~to~} (\tilde{\Delta}_ib,v'')\\
	\forall & v',v''  \in V_i
	\end{aligned}
	\right.
	\right\}\right].
	\end{equation}
	Thus, by throughput definition \eqref{eqn:throughput-def}, 
	\begin{equation}\label{eqn:pf2-1}
	\theta(N_i)\ge
	\left(\frac{1}{\tilde{\Delta}_ib}\mathbb{E}\left[\max_{\pi}\!\left\{Wei(\pi)\left|
	\begin{aligned}
	\pi & \text{~from~}  (2\tilde{\Delta}_ib,v')\\
	& \text{~to~} (\tilde{\Delta}_ib,v'')\\
	\forall & v',v''  \in V_i
	\end{aligned}
	\right.
	\right\}\right]\right)^{-1}.
	\end{equation}
	By Lemma \ref{lem:exist-path}, there always exists a path from $(3\tilde{\Delta}_ib,v)$ to $(2\tilde{\Delta}_{i}b,v')$ and a path from $(\tilde{\Delta}_{i}b,v'')$ to $(0,v)$, $\forall v,v',v''\in V_i$. Thus, we have 
	\begin{eqnarray}
		&&\hspace{-0.3in}\max_{\pi}\left\{Wei(\pi)|\pi:(2\tilde{\Delta}_ib,v')\leadsto (\tilde{\Delta}_ib,v''),~\forall v',v''\in V_i \right\}\nonumber\\
		&\hspace{-0.1in}=\hspace{-0.1in}&
		\max_{v',v''\in V_i}\left\{Wei(\pi^*_{(2\tilde{\Delta}_ib,v')\leadsto (\tilde{\Delta}_ib,v'')})\right\}
		\nonumber\\	
		&\hspace{-0.1in}\le \hspace{-0.1in}&\max_{v',v''\in V_i}\left\{Wei(\pi^*_{(3\tilde{\Delta}_ib,\tilde{v})\leadsto (2\tilde{\Delta}_ib,v')})-S_{2\tilde{\Delta}_ib,v'}(N_i)\right.\nonumber\\
		&&\hspace{0.5in}\left.+Wei(\pi^*_{(2\tilde{\Delta}_ib,v')\leadsto (\tilde{\Delta}_ib,v'')})
		\right.\nonumber\\
		&&\hspace{0.5in}\left.+Wei(\pi^*_{(\tilde{\Delta}_ib,v'')\leadsto (0,\tilde{v})})-S_{\tilde{\Delta}_ib,v'}(N_i)\right\}
		\nonumber\\
		&\hspace{-0.1in}\le \hspace{-0.1in}&\max_{\pi}\left\{Wei(\pi)|\pi:(3\tilde{\Delta}_ib,\tilde{v})\leadsto (0,\tilde{v})\right\},\label{eqn:pf2-2}
	\end{eqnarray}
	where the last inequality is by the supper additive property of the maximum weighted path (see Lemma \ref{lem:super-additive}). 
	Let 
	\begin{equation}
		\pi^*_{3\tilde{\Delta}_ib,\tilde{v}}:=\argmax_{\pi}\left\{Wei(\pi)|\pi:(3\tilde{\Delta}_ib,\tilde{v})\leadsto (0,\tilde{v})\right\}.
	\end{equation}
	Combining \eqref{eqn:pf2-1} and \eqref{eqn:pf2-2} yields
	\begin{equation}\label{eqn:pf2-3}
		\liminf_{i\rightarrow \infty} \theta(N_i)
		\ge \liminf_{i\rightarrow \infty}\left(\mathbb{E}\left[\frac{Wei\big(\pi^*_{3\tilde{\Delta}_ib,\tilde{v}}\big)}{\tilde{\Delta}_ib}\right]\right)^{-1}.
	\end{equation}
	
	\bigskip
	\noindent {\it Part b: showing the lower bound is strictly positive by mapping networks onto lattices.}

	Since $dim_M(G_i)\le K-1$, there must exists a resolving set $W_i$ with cardinality $k_i\le K-1$ for $G_i$. By definition of the resolving set, we can embed $\mathcal{V}_i$ into a $K$-dimensional lattice by mapping each node $(m,v)\in \mathcal{V}_i$ to a unique point $\big(m,\bm{r}(v|W_i)-\bm{r}(\tilde{v}|W_i)\big)\in \mathcal{L}^{K}$, where zero elements are added if $k_i<K-1$. 
	Denote this one-to-one mapping as $\mathbf{M}$: 
	\begin{equation}
		\mathbf{M}(m,v)=\big(m,\bm{r}(v|W_i)-\bm{r}(\tilde{v}|W_i)\big). 
	\end{equation}

	Consider a path $\pi_{3\tilde{\Delta}_ib,\tilde{v}}:(3\tilde{\Delta}_ib,\tilde{v})\leadsto(0,\tilde{v})$. Let $|\pi_{3\tilde{\Delta}_ib,\tilde{v}}|$ denote the number of arcs on $\pi_{3\tilde{\Delta}_ib,\tilde{v}}$. 
	Let $|\pi^{I}_{3\tilde{\Delta}_ib,\tilde{v}}|$, $|\pi^{II}_{3\tilde{\Delta}_ib,\tilde{v}}|$, $|\pi^{III}_{3\tilde{\Delta}_ib,\tilde{v}}|$ denote the number of Type I, Type II, Type III arcs on $\pi_{3\tilde{\Delta}_ib,\tilde{v}}$, respectively. By Lemma \ref{lem:path-upper-bound}, we have 
	\begin{equation}
		\left|\pi_{3\tilde{\Delta}_ib,\tilde{v}}\right|\le \max\{L^*_i+1,b\}\cdot 3\tilde{\Delta}_i.
	\end{equation}
	\begin{equation}
		\left|\pi^{III}_{3\tilde{\Delta}_ib,\tilde{v}}\right|\le 3\tilde{\Delta}_i.
	\end{equation}
	
	Next, we show that the nodes on any $\pi_{3ib,\tilde{v}}$, after the one-to-one mapping $\mathbf{M}$, can be covered by a lattice animal $\xi$ 
	on a $K$-dimensional lattice. We adopt the concepts of lattice and lattice animal from \cite{Cox1993,Martin-LatticeAnimal}. A $K$-dimensional lattice, denoted as $\mathcal{L}^K$, is a graph of $\mathbb{Z}^K$. Two points $\bm{x}$ and $\bm{y}$ on $\mathcal{L}^K$ are adjacent if and only if $|\bm{x}-\bm{y}|=1$. A set of points on $\mathcal{L}^K$ is connected if and only if any pair of points in the set can be connected by a sequence of adjacent points within the set. A $K$-dimensional {\it lattice animal} $\xi$ is defined as a finite connected subset of a $K$-dimensional lattice.

	For each Type I arc $(m,v)\rightarrow(m,u)$ in $\mathcal{E}_i$ with $(u,v)\in E_i$, we have
	\begin{eqnarray}
		&&||\mathbf{M}(m,v)-\mathbf{M}(m,u)||_\infty\nonumber\\
		&=&
		||\big(m,\bm{r}(v|W_i)-\bm{r}(\tilde{v}|W_i)\big)-\big(m,\bm{r}(u|W_i)-\bm{r}(\tilde{v}|W_i)\big)||_\infty\nonumber\\
		&=&||\bm{r}(v|W_i)-\bm{r}(u|W_i)||_\infty\nonumber\\
		&=& \max_{w\in W_i}\{|dis(v,w)-dis(u,w)|\}\nonumber\\
		&\le & dis(u,v)=1.
	\end{eqnarray}
	Thus, the two points $\mathbf{M}(m,v)$ and $\mathbf{M}(m,u)$ can be connected by at most $K-1$ intermediate points (when diagonal) in the $K$-dimensional lattice. 
	
	For each Type II arc $(m,v)\rightarrow(m-1,v)$ in $\mathcal{E}_i$, the two points $\mathbf{M}(m,v)$ and $\mathbf{M}(m-1,v)$ is directly adjacent in $\mathcal{L}^K$. 
	
	For each Type III arc $(m,v)\rightarrow(m-b,u)$ in $\mathcal{E}_i$, the two points $\mathbf{M}(m,v)$ and $\mathbf{M}(m-b,u)$ can be connected by adjacent intermediate points $\mathbf{M}(m-1,v),\mathbf{M}(m-2,v),\dots,\mathbf{M}(m-b,v)$ along with at most $K-1$ points from $\mathbf{M}(m-b,v)$ to $\mathbf{M}(m-b,u)$ in $\mathcal{L}^K$, as $(v,u)\in E_i$. 
	
	In summary, the nodes on $\pi_{3\tilde{\Delta}_ib,\tilde{v}}$ and the extra adjacent intermediate nodes for Type I and Type III arcs on $\pi_{3\tilde{\Delta}_ib,\tilde{v}}$ together form a lattice animal $\xi$ on $\mathcal{L}^K$ that covers all nodes on $\pi_{3\tilde{\Delta}_ib,\tilde{v}}$. The size of such lattice animal is bounded above by
	\begin{eqnarray}
		\left|\xi\right|&\le& K\left|\pi^{I}_{3\tilde{\Delta}_ib,\tilde{v}}\right|+\left|\pi^{II}_{3\tilde{\Delta}_ib,\tilde{v}}\right|+(b+K)\left|\pi^{III}_{3\tilde{\Delta}_ib,\tilde{v}}\right|+1
		\nonumber\\
		&\le&
		K\max\{L^*_i+1,b\}\cdot 3\tilde{\Delta}_i+3b\tilde{\Delta}_i+1,
	\end{eqnarray} 
	where the term $+1$ is by considering the point $(0,\bm{0})$ where the path ends. Since $\limsup_{i\rightarrow\infty}L^*_i\!<\!\infty$, there must exist a constant $c<\infty$ such that $L^*_i<c$ for all $i\in \mathbb{Z}^+$. Thus,
	\begin{equation}
		\left|\xi\right|\le K\max\{c+1,b\}\cdot 3\tilde{\Delta}_i+3b\tilde{\Delta}_i+1=:f(\tilde{\Delta}_i)\in \mathbb{Z}^+
	\end{equation}
	This suggests that the nodes on any path $\pi_{3\tilde{\Delta}_ib,\tilde{v}}$, after the one-to-one mapping $\mathbf{M}$, can be covered by a lattice animal $\xi$ on $\mathcal{L}^K$ of size $f(\tilde{\Delta}_i)$ containing $(0,\bm{0})$. Then $Wei\big(\pi_{3\tilde{\Delta}_ib,\tilde{v}}\big)$ is bounded above by $Wei(\xi)$, where $Wei(\xi):=\sum_{(m,\bm{r})\in \xi}S_{m,\bm{r}}$ is the weight of $\xi$ and $S_{m,\bm{r}}$ follows i.i.d. $F_\sigma$ for all $(m,\bm{r})$ on $\mathcal{L}^K$. 
	
	Consequently, the weight of the maximum weighted path $\pi^*_{3\tilde{\Delta}_ib,\tilde{v}}$ is bounded above by the weight of a ``greedy lattice animal" on $\mathcal{L}^K$ of size $f(\tilde{\Delta}_i)$ containing $(0,\bm{0})$, i.e.
	\begin{equation}\label{eqn:greedy-lattice-animal-bound}
		Wei\big(\pi^*_{3\tilde{\Delta}_ib,\tilde{v}}\big)\le \max_{\xi\in \mathcal{A}^K\big(f(\tilde{\Delta}_i)\big)}Wei(\xi),
	\end{equation}
	where $\mathcal{A}^K\big(f(\tilde{\Delta}_i)\big)$ is the set of all lattice animals on $\mathcal{L}^K$ of size $f(\tilde{\Delta}_i)$ containing~$(0,\bm{0})$. Thus, we have
	\begin{eqnarray}
		\liminf_{i\rightarrow \infty} \theta(N_i)
		&\!\!\!\!\ge\!\!\!\!& \liminf_{i\rightarrow \infty}\left(\mathbb{E}\!\left[\frac{\max_{\xi\in \mathcal{A}^K(f(\tilde{\Delta}_i))}Wei(\xi)}{\tilde{\Delta}_ib}\right]\right)^{-1}\nonumber\\
		&\!\!\!\!=\!\!\!\!&\!
		\left(\!\limsup_{i\rightarrow \infty}\mathbb{E}\!\left[\frac{\max_{\xi\in \mathcal{A}^K(f(\tilde{\Delta}_i))}Wei(\xi)}{f(\tilde{\Delta}_i)}\!\cdot\!\frac{f(\tilde{\Delta}_i)}{\tilde{\Delta}_ib}\right]\!\right)^{-1}
		.
	\end{eqnarray}
	Since $\tilde{\Delta}_i$ diverges monotonically to infinite as $i\rightarrow\infty$, we have 
	
	$\bullet$ $\{ f(\tilde{\Delta}_i)\}_{i=1}^\infty$ is monotone;
	
	$\bullet$ $\lim_{i\rightarrow\infty} f(\tilde{\Delta}_i)=\infty$;
	
	$\bullet$ $\lim_{i\rightarrow\infty} \frac{f(\tilde{\Delta}_i)}{\tilde{\Delta}_ib}=3K\max\{(c+1)b,1\}+3<\infty$.
	
	\noindent Thus, 
	\begin{eqnarray}
		&&\liminf_{i\rightarrow \infty} \theta(N_i)\nonumber\\
		&\ge&\left(c'\cdot \lim_{i'\rightarrow \infty}\sup_{~i\ge i'}\mathbb{E}\left[\frac{\max_{\xi\in \mathcal{A}^K(f(\tilde{\Delta}_i))}Wei(\xi)}{f(\tilde{\Delta}_i)}\right]\right)^{-1}\nonumber\\
		&\ge&\left(c'\cdot\lim_{i'\rightarrow \infty}\sup_{~~f\ge f(\tilde{\Delta}_{i'}),f\in \mathbb{Z}^+}\mathbb{E}\left[\frac{\max_{\xi\in \mathcal{A}^K(n)}Wei(\xi)}{f}\right]\right)^{-1}\nonumber\\
		&=&\left(c'\cdot\!\!\!\lim_{~~f(\tilde{\Delta}_{i'})\rightarrow \infty}\sup_{f\ge f(\tilde{\Delta}_{i'}),f\in \mathbb{Z}^+}\mathbb{E}\left[\frac{\max_{\xi\in \mathcal{A}^K(f)}Wei(\xi)}{f}\right]\right)^{-1}\nonumber\\
		&=&\left(c'\cdot \limsup_{n\rightarrow\infty}\mathbb{E}\left[\frac{\max_{\xi\in \mathcal{A}^K(n)}Wei(\xi)}{n}\right]\right)^{-1},
	\end{eqnarray}
	where $c'=3K\max\{(c+1)b,1\}+3$ is a constant. 
	By \cite[Theorem 2.3]{Martin-LatticeAnimal}, under Condition \eqref{eqn:FsigmaK}, 
	there exists a constant $c''<\infty$ such that
	\begin{equation}\label{eqn:Martin-result}
		\sup_{n}\mathbb{E}\left[\frac{\max_{\xi\in \mathcal{A}^K(n)}Wei(\xi)}{n}\right]\le c''\cdot\int_0^\infty \big(1-F_{\sigma}(x)\big)^{1/K}dx<\infty.
	\end{equation}
	Thus, 
	\begin{equation}
		\liminf_{i\rightarrow \infty} \theta(N_i)\ge \left(c'\cdot c''\cdot\int_0^\infty \big(1-F_{\sigma}(x)\big)^{1/K}dx\right)^{-1}>0.
	\end{equation}
\end{proof}
\smallskip

Since $dim_{EM}(\mathcal{N})\le K-1$, there must exists a constant $\Lambda>0$ such that $\forall i\in \mathbb{Z}^+$, the graph $G_i$ has a $\Lambda$-extended resolving set $\mathcal{W}_i$ with cardinality $k_i\le K-1$.
Thus, $\forall i\in \mathbb{Z}^+$, we can embed $\mathcal{V}_i$ (the node set of the precedence graph $\mathcal{P}_i$) into a $K$-dimensional lattice by mapping no more than $\Lambda$ nodes $(m,v)\in\mathcal{V}_i$ to a unique point $\big(m,\bm{r}(v|\mathcal{W}_i)-\bm{r}(\tilde{v}|\mathcal{W}_i)\big)\in \mathcal{L}^K$, where zero elements are added if $k_i<K-1$. Denote this many-to-one mapping as $\mathbf{M}^\dagger$:
\begin{equation}\label{eqn:many-to-one-mapping}
	\mathbf{M}^\dagger(m,v)=\big(m,\bm{r}(v|\mathcal{W}_i)-\bm{r}(\tilde{v}|\mathcal{W}_i)\big). 
\end{equation}

Consider a path $\pi_{3\tilde{\Delta}_ib,\tilde{v}}:(3\tilde{\Delta}_ib,\tilde{v})\leadsto(0,\tilde{v})$ and consider the set of nodes on the path. Let 
\begin{equation}\vspace{-0.05in}
	\mathcal{L}(\pi_{3\tilde{\Delta}_ib,\tilde{v}})=\bigcup_{(m,v)\in \pi_{3\tilde{\Delta}_ib,\tilde{v}}} \{\mathbf{M}^\dagger(m,v)\}
\end{equation}
be the set of points in $\mathcal{L}^K$ that the nodes on $\pi_{3\tilde{\Delta}_ib,\tilde{v}}$ maps to. Because of the many-to-one mapping, each point $(m,\bm{r})\in \mathcal{L}^K$ could be visited multiple times along the path~$\pi_{3\tilde{\Delta}_ib,\tilde{v}}$. However, since each node $(m,v)\in\mathcal{V}_i$ can be visited only once by~$\pi_{3\tilde{\Delta}_ib,\tilde{v}}$ (as the precedence graph is acyclic by \cite[Lemma 5.1]{Dallery}) and no more than $\Lambda$ nodes in~$\mathcal{V}_i$ can be mapped to the same point in $\mathcal{L}^K$, each point $(m,\bm{r})\in \mathcal{L}^K$ can be visited at most $\Lambda$ many times by $\pi_{3\tilde{\Delta}_ib,\tilde{v}}$. Thus, $Wei(\pi_{3\tilde{\Delta}_ib,\tilde{v}})$ is bounded above by counting the weight of all points in $\mathcal{L}(\pi_{3\tilde{\Delta}_ib,\tilde{v}})$ for $\Lambda$ many times independently, where the weight of each point in $\mathcal{L}(\pi_{3\tilde{\Delta}_ib,\tilde{v}})$ follows $F_{\sigma}$. Alternatively, we can let the weight of each point in  $\mathcal{L}^K$ follows $F_{\sigma}^{*\Lambda}$, i.e. the $\Lambda$-fold convolution of $F_{\sigma}$, and then $Wei(\pi_{3\tilde{\Delta}_ib,\tilde{v}})$ is bounded above by the weight of all points in $\mathcal{L}(\pi_{3\tilde{\Delta}_ib,\tilde{v}})$. 

For any $(u,v)\in E_i$ and for any $W\in \mathcal{W}_i$, let 
\begin{equation}\vspace{-0.1in}
	u^*=\argmin_{w\in W}\{dis(u,w)\},
\end{equation}
\begin{equation}
	v^*=\argmin_{w\in W}\{dis(v,w)\}. 
\end{equation}
We have
\begin{eqnarray}
	&&|dis(v,W)-dis(u,W)|\nonumber\\
	&=&|dis(v,v^*)-dis(u,u^*)|\nonumber\\
	&=&\max\{dis(v,v^*)-dis(u,u^*),dis(u,u^*)-dis(v,v^*)\}\nonumber\\
	&\le & \max\{dis(v,u^*)-dis(u,u^*),dis(u,v^*)-dis(v,v^*)\}\nonumber\\
	&\le & dis(v,u).
\end{eqnarray}
Thus, 
\begin{eqnarray}
	&&||\mathbf{M}^\dagger(m,v)-\mathbf{M}^\dagger(m,u)||_\infty\nonumber\\
	&=&
	||\big(m,\bm{r}(v|\mathcal{W}_i)-\bm{r}(\tilde{v}|\mathcal{W}_i)\big)-\big(m,\bm{r}(u|\mathcal{W}_i)-\bm{r}(\tilde{v}|\mathcal{W}_i)\big)||_\infty\nonumber\\
	&=&||\bm{r}(v|\mathcal{W}_i)-\bm{r}(u|\mathcal{W}_i)||_\infty\nonumber\\
	&=& \max_{W\in \mathcal{W}_i}\{|dis(v,W)-dis(u,W)|\}\nonumber\\
	&\le & dis(u,v)= 1.
\end{eqnarray}
Similar to the proof of Lemma \ref{lem:scal-suf-metric-d}, we can cover the nodes in $\mathcal{L}(\pi_{3\tilde{\Delta}_ib,\tilde{v}})$ by a lattice animal $\xi$ on $\mathcal{L}^K$ of size $f(\tilde{\Delta}_i):=K\max\{c+1,b\}\cdot 3\tilde{\Delta}_i+3b\tilde{\Delta}_i+1$ containing $(0,\bm{0})$. Thus, $Wei\big(\pi_{3\tilde{\Delta}_ib,\tilde{v}}\big)$ is bounded above by $Wei(\xi)=\sum_{(m,\bm{r})\in \xi}S_{m,\bm{r}}$, where $|\xi|=f(\tilde{\Delta}_i)$, $(0,\bm{0})\in \xi$, and $S_{m,\bm{r}}$ follows i.i.d. $F_\sigma^{*\Lambda}$. 

As $\int_0^\infty \big(1-F_{\sigma}(x)\big)^{1/K}dx<\infty$ and $\Lambda>0$ is a constant, we have 
\begin{eqnarray}\label{eqn:Lambda-convolution-tail}
	\int_0^\infty \big(1-F_{\sigma}^{*\Lambda}(x)\big)^{1/K}dx \!\!\!
	&=&\!\!\!\int_0^\infty \mathbb{P}\{\sigma_1+\dots+\sigma_{\Lambda}>x \}^{1/K}dx  \nonumber\\
	\!\!\!&\le &\!\!\! \int_0^\infty\left(\sum_{n=1}^\Lambda\mathbb{P}\left\{\sigma_n>\frac{x}{\Lambda}\right\}\right)^{1/K}\!\!\!dx  \nonumber\\
	\!\!\!&=&\!\!\! (\Lambda)^{1/K}\Lambda\int_0^\infty \big(1-F_{\sigma}(x)\big)^{1/K}dx
	\nonumber\\
	\!\!\!&<&\!\!\!\infty,
\end{eqnarray}
where $\sigma_1,\sigma_2,\dots,\sigma_{\Lambda}$ follow i.i.d. $F_{\sigma}$. The rest of the proof follows as the proof of Lemma \ref{lem:scal-suf-metric-d}.

\section{Acknowledgments}
	This work was supported by the National Science Foundation under grants CNS-1717060, IIS-0916440, ECCS-1232118, SES-1409214. 

\bibliographystyle{abbrv}
\bibliography{FJQN-B-heavytail-bib}

\end{document}